\documentclass[twocolumn, aps, prl, nofootinbib, superscriptaddress]{revtex4}
\usepackage{graphicx}
\usepackage{amsfonts}
\usepackage{amssymb}
\usepackage{amsthm}
\usepackage{array}
\usepackage{amsmath}
\usepackage{verbatim} 
\usepackage{hyperref}
\usepackage{color}
\usepackage{bbold}
\usepackage{epstopdf}
\usepackage{mathtools}
\usepackage[normalem]{ulem}

\newtheorem{lemma}{Lemma}
\newtheorem{proposition}{Proposition}

\newtheorem{theorem}{Theorem}

\newcommand{\ket}[1]{\left\vert#1\right\rangle}
\newcommand{\bra}[1]{\left\langle#1\right\vert}

\def\bra#1{\langle #1|}
\def\ket#1{ |#1 \rangle}

\def\Tr{\mbox{\rm Tr}}

\begin{document}
\title{Condition on the R\'enyi Entanglement Entropy under Stochastic Local Manipulation}
\author{Hyukjoon Kwon}
\affiliation{QOLS, Blackett Laboratory, Imperial College London, London SW7 2AZ, United Kingdom}
\author{A. J. Paige}
\affiliation{QOLS, Blackett Laboratory, Imperial College London, London SW7 2AZ, United Kingdom}
\author{M. S. Kim}
\affiliation{QOLS, Blackett Laboratory, Imperial College London, London SW7 2AZ, United Kingdom}
\affiliation{Korea Institute for Advanced Study, Seoul, 02455, South Korea}
\begin{abstract}
The R\'enyi entanglement entropy (REE) is an entanglement quantifier considered as a natural generalisation of the entanglement entropy. When it comes to stochastic local operations and classical communication (SLOCC), however, only a limited class of the REEs satisfy the monotonicity condition, while their statistical properties beyond mean values have not been fully investigated. Here, we establish a general condition that the probability distribution of the REE of any order obeys under SLOCC. The condition is obtained by introducing a family of entanglement monotones that contain the higher-order moments of the REEs. The contribution from the higher-order moments imposes a strict limitation on entanglement distillation via SLOCC. We find that the upper bound on success probabilities for entanglement distillation exponentially decreases as the amount of raised entanglement increases, which cannot be captured from the monotonicity of the REE. Based on the strong restriction on entanglement transformation under SLOCC, we design a new method to estimate entanglement in quantum many-body systems from experimentally observable quantities.
\end{abstract}
\maketitle

Quantum entanglement is an essential resource to achieve quantum advantages in various nonclassical tasks, including quantum teleportation \cite{Bennet93} and communication \cite{Horodecki05}. The fields of many-body physics have recognised entanglement as a useful quantity to characterise quantum ground states \cite{Eisert10} and to witness quantum phase-transitions \cite{Osterloh02, Bayat17}. For a bipartite pure state $\ket{\Psi}_{AB}$, the most widely-studied measure to quantify entanglement is the entanglement entropy, $E_S (\Psi) = S(\rho_B) = - \Tr \left[ \rho_B \log \rho_B \right]$, given by the von Neumann entropy of  the local quantum state $\rho_{B} = \Tr_{A} \ket{\Psi}_{AB} \bra{\Psi} $. More generally, the R\'enyi entanglement entropy (REE) \cite{Vidal00},
\begin{equation}
\label{Eq:ReeDef}
E_\alpha (\Psi) = S_\alpha (\rho_B) = \frac{1}{1-\alpha} \log \Tr[\rho^\alpha_B],
\end{equation}
has been studied as an extended class of entanglement quantifiers. The entanglement entropy $E_S$ can be retrieved in the limit $\alpha \rightarrow 1$. The REEs of low and high $\alpha$'€™s behave differently given changes of the entanglement spectrum \cite{Li08, Calbrese08}, which has made them useful to classify quantum phases \cite{Flammia09, Cui12, Franchini14}.  The REE and R\'enyi entropy of order $\alpha <1$, especially $\alpha = 1/2$, have been attracting attention for quantifying quantum correlations in many-body systems \cite{Calabrese13, Alba19, Camilo19}. Meanwhile, the REE of order $\alpha > 1$ can be estimated without quantum-state tomography \cite{Ekert02, Calabrese04, Mintert07}, and for $\alpha =2$, it has recently been measured experimentally in quantum many-body systems \cite{Islam15, Kaufman16, Linke18, Brydges19}.

One of the most important properties of entanglement is that we cannot increase it deterministically by any local operation and classical communication (LOCC) protocols. Nevertheless, it is possible to distill the maximally entangled state from a partially entangled state by allowing the success probability to be less than unity \cite{Jonathan99, Lo01}. The monotonicity condition for entanglement measures ensures that the degree of average entanglement for the total system does not increase by any stochastic LOCC (SLOCC) protocols. The REEs of order $0 \leq \alpha \leq 1$ satisfy such the condition \cite{Vidal00}, and they serve as reliable measures of entanglement. However, despite the advantages to be experimentally measurable, the REEs of order $\alpha > 1$ have a limitation as they do not satisfy the monotonicity condition \cite{Horodecki09}. Here, recalling that the monotonicity condition is based on the average entanglement, we ask the question whether the REEs of any order can be of use to characterise entanglement under SLOCC by taking into account their higher-order moments. There have been attempts to find refined conditions on the statistical properties of entanglement beyond its mean value \cite{Jonathan99, Lo01}. These have, however, been limited to the study of the success probability of nondeterministic transformations when the exact form, i.e., the Schmidt decompositions of outcome states, is given.

In this Letter, we explore a general condition on entanglement transformation through SLOCC by focusing only on the REE of the outcome states, without characterising their Schmidt decompositions. We introduce a generalised entanglement entropy (GEE) as an entanglement monotone, and based on this, establish a condition on the distribution of the REE under any SLOCC protocols. From this condition, we demonstrate that the success probability of raising entanglement exponentially decreases as the entanglement required to distill increases. This provides a strong limitation on entanglement manipulation under SLOCC protocols, for instance in distilling a moderate amount of entanglement. Our results can be applied to the estimation of entanglement in quantum many-body systems, as a lower bound on the REE is obtained from the higher-order moments of $E_2$ after applying an SLOCC protocol. Finally, we discuss how our results can be extended for mixed states.

{\it Condition on the REE distribution under SLOCC}.---Let us suppose that the initial bipartite pure state $\ket{\Psi}_{AB}$ transforms through an SLOCC protocol as 
\begin{equation}
\label{Eq:LOCC}
\ket{\Psi}_{AB} \xrightarrow{{\cal E}_{\rm SLOCC}} \{ p_m, \ket{\Psi_m}_{AB} \},
\end{equation}
with outcome probabilities $0 \leq p_m \leq 1$ satisfying $\sum_m p_m = 1$.
Entanglement of the outcome state can increase depending on the SLOCC protocol, while the monotonicity of the REE of order $0 \leq \alpha \leq 1$ guarantees that average entanglement does not increase through any SLOCC protocol, i.e., $\langle \Delta E_\alpha \rangle \leq 0$. Here, $\langle O \rangle := \sum_m p_m O(\Psi_m)$ and $\Delta E_{\alpha}(\Psi_m) := E_{\alpha} (\Psi_m) - E_{\alpha}(\Psi)$. However, the REE of order $\alpha > 1$ does not obey such condition \cite{Horodecki09}. Nevertheless, we find the following statistical property that holds for the REE of any order.
\begin{theorem}
\label{Thm:REE}
Under any SLOCC protocol, the outcome statistics of the REE obey the following:
\begin{equation}
\label{Eq:22}
\langle e^{ s (1-\alpha) \Delta E_{\alpha}} \rangle 
\begin{cases}
  \leq 1 & (0 < \alpha < 1~{\rm and}~  s \leq \frac{1}{\alpha} ) \\
  \geq 1 & ( \alpha > 1~{\rm and}~ s \geq \frac{1}{\alpha}) \qquad.
\end{cases}
\end{equation}
\end{theorem}
\noindent We denote $\Omega := \{ (\alpha, s) | 0 < \alpha < 1~{\rm and}~  s \leq 1/\alpha \} \cup \{ (\alpha, s) | \alpha > 1~{\rm and}~ s \geq 1/\alpha \}$, where Eq.~\eqref{Eq:22} holds. We note that the condition is stronger than the monotonicity condition of the REEs of order $0<\alpha < 1$ due to the convexity of the exponential function. It also provides more information about the higher-order moments $\langle (\Delta E_\alpha)^k \rangle$ beyond the mean value for any $\alpha \neq 1$. This result can be compared to entanglement fluctuation theorems \cite{Alhambra19, Kwon19} showing the equality relations regarding the higher-order moments of outcome entanglement statistics. While previous works focus on a specific LOCC protocol \cite{Alhambra19} and the fluctuation of the Schmidt coefficients \cite{Kwon19}, our results do not depend on the form of LOCC protocol while dealing with the spectrum of entanglement measures given by the REE.

We sketch the proof of Theorem~\ref{Thm:REE} by introducing a family of entanglement monotones, GEEs, based on the generalised quantum entropy \cite{Hu06, Rastegin11, Bourin11}.
\begin{proposition} 
\label{Prop:AntiNorm}
The GEE defined as
\begin{equation}
E_{(\alpha, s)} (\Psi) := \frac{1}{s(1-\alpha)} \left[ e^{ s(1-\alpha) E_\alpha (\Psi) } -1 \right]
\end{equation}
is an entanglement monotone for $(\alpha, s) \in \Omega$, satisfying the following conditions:
(i) $E_{(\alpha,s)} (\Psi) = 0$ \textit{if and only if} $\ket{\Psi}_{AB}$ is separable, and (ii) $\langle \Delta E_{(\alpha,s)} \rangle \leq  0$ via SLOCC. When $s \rightarrow 0$, $E_{(\alpha,s)}$ becomes $E_\alpha$, whereas $s=1$ gives the entanglement measure based on the Tsallis entropy \cite{Tsallis88}.
\end{proposition}
\noindent The monotonicity of the GEE can be obtained from the concavity of the generalised entropy on positive semi-definite matrices \cite{Hu06, Rastegin11, Bourin11} combined with Vidal's work \cite{Vidal00}, from which Theorem~\ref{Thm:REE} naturally follows. Detailed proofs of the Theorem and Propositions throughout this Letter can be found in the Supplemental Material \cite{Suppl}.

We focus on the case $s=1/\alpha$, where $E_{(\alpha, s)}$ becomes an entanglement monotone for all $\alpha \in (0, 1) \cup (1, \infty)$. When $\alpha$ approaches $1$, the measure converges to the entanglement entropy, i.e., $\displaystyle \lim_{\alpha \rightarrow 1} E_{(\alpha,1/\alpha)} = E_S$, and Eq.~\eqref{Eq:22} becomes $\langle \Delta E_S \rangle \leq 0$. We also note that $E_{(\alpha, 1/\alpha)}$ for $\alpha = 1/2$ has a direct connection to the entanglement negativity for pure bipartite states as $E_{(1/2,2)} (\Psi)  = \| (\ket\Psi_{AB} \bra\Psi)^{T_B} \|_1 -1 $, where $T_B$ denotes the partial transpose on $B$ \cite{Vidal02}. We will show that improved bounds on entanglement distillation and estimation tasks can be obtained from Theorem~\ref{Thm:REE}, and $s=1/\alpha$ gives the tightest bounds for both cases.

{\it Probability bound on entanglement distillation via SLOCC}.---We first demonstrate that Theorem~\ref{Thm:REE} leads to a strong restriction on entanglement distillation via SLOCC. Let us clarify the problem by defining the accumulated success probability $P \left(  E_\alpha  \geq E_{\rm target} \right) = \sum_{m \in \Gamma}  p_m$, where $\Gamma = \left\{ m |  E_\alpha (\Psi_m)  \geq E_{\rm target} \right\}$ is a set where the outcome states $\ket{\Psi_m}_{AB}$ have entanglement $E_\alpha(\Psi_m)$ higher than the desired value $E_{\rm target}$.

Finding the highest $P \left(  E_\alpha  \geq E_{\rm target} \right)$ can be reformulated by using the necessary and sufficient condition for a pure state transition under SLOCC \cite{Jonathan99}. Furthermore, for \textit{any entanglement monotone} $E$, we show that the optimal success probability is given by
$$
\begin{aligned}
&\sup_{{\cal E}_{\rm SLOCC}} P \left(  E  \geq E_{\rm target} \right) \\
&~ = \max_{\Psi'} \left[ \min_{l \in \{ 1,2, \cdots, d\} } \left( \frac{\sum_{i=l}^d \lambda_i^\downarrow(\Psi)} {\sum_{i=l}^d \lambda_i^\downarrow(\Psi')} \right) \vline E(\Psi') = E_{\rm target} \right],
\end{aligned}
$$
where $\lambda^\downarrow_{i} (\Psi)$ are the Schmidt coefficients of $\ket{\Psi}_{AB}$ in descending order with the Schmidt rank $d$ (see the Supplemental Material \cite{Suppl} for the proof). However, this optimisation problem is computationally challenging as it does not belong to a convex optimisation problem \cite{note1}. Evaluating the highest success probability requires searching all possible outcome states satisfying the constraint, $E(\Psi') = E_{\rm target}$, and the nonlinearity of entanglement monotones, including the REE, makes the problem even more complicated especially for large $d$. Furthermore, no such optimisation can be applied to $E_\alpha$ for $\alpha >1$, since it is not an entanglement monotone.

Instead of finding the exact solution to this problem, we investigate an upper bound on the success probability for entanglement distillation. We find the following bound based on Theorem~\ref{Thm:REE}, or equivalently the monotonicity of the GEE:
\begin{proposition} 
$\forall \alpha >0$, the accumulated success probability for entanglement distillation is upper bounded by
\begin{equation}
\label{Eq:PalphaBound}
P\left(  E_\alpha  \geq E_{\rm target}\right) \leq \frac{ e^{ \left( \frac{1-\alpha}{\alpha} \right) E_\alpha(\Psi) } -1 } { e^{  \left( \frac{1-\alpha}{\alpha} \right) E_{\rm target} }-1} =: P^{\rm GEE}_\alpha.
\end{equation}
\end{proposition}
\noindent For $0 < \alpha < 1$, this bound is tighter than the probability bound derived from the monotonicity of the REE, 
\begin{equation}
\label{Eq:MonotoneBound}
\langle \Delta E_\alpha \rangle \leq 0 \Longrightarrow P\left(  E_\alpha  \geq E_{\rm target} \right) \leq \frac{E_\alpha(\Psi)}{E_{\rm target}} =: P_\alpha^{\rm REE}.
\end{equation}
We further extend the bound for $\alpha = 0$ and $\alpha = 1$. For $E_0(\Psi) = \log d < E_{\rm target}$, $P_\alpha^{\rm GEE}$ reaches zero when $\alpha$ approaches zero, implying that no SLOCC can increase the Schmidt rank. For the distillation of the maximally entangled state $\ket{\Phi_d}_{AB} \propto \sum_{i=1}^d \ket{ii}_{AB}$, i.e., $E_{\rm target} = \log d$, we have $\displaystyle \lim_{\alpha \rightarrow 0}P_\alpha^{\rm GEE} = d e^{ (1/d) {\rm Tr} [\log \rho_B]} $, which coincides with the bound given by the $G$-concurrence monotone \cite{Gour05}. Although this bound does not reach the optimal rate found in Ref.~\cite{Jonathan99}, it can be obtained with less information without knowing all the Schmidt coefficients. When $\alpha$ approaches $1$, the bound reduces to $\displaystyle \lim_{\alpha \rightarrow 1} P_\alpha^{\rm GEE} ={E_S(\Psi)}/{E_{\rm target}}$.

By noting that $E_\alpha \leq E_\beta$ for $\alpha \geq \beta$, Eq.~\eqref{Eq:PalphaBound} can be extended to bounds between the REEs of different orders, $P \left(  E_\alpha  \geq E_{\rm target} \right) \leq P \left(  E_\beta  \geq E_{\rm target} \right) \leq P_\beta^{\rm GEE}$. For example, the success probability for raising the entanglement entropy can be bounded as $P(E_S> E_{\rm target})  \leq P_{1/2}^{\rm GEE}$, which can be expressed in terms of the entanglement negativity.
Optimising over all possible $\beta$ leads to
\begin{equation}
\label{Eq:PBoundOpt}
P \left(  E_\alpha  \geq E_{\rm target} \right) \leq \min_{0 \leq \beta \leq \alpha} P_\beta^{\rm GEE} =:P_{\alpha, {\rm opt}}^{\rm GEE}.
\end{equation}
This shows a strong limitation for entanglement distillation via SLOCC as the upper bound on success probabilities exponentially decreases as 
$P( \Delta E_\alpha \geq x) = P( E_\alpha \geq E_\alpha(\Psi) + x) \leq e^{- \left( \frac{1-\alpha}{\alpha} \right) x}$
for $0< \alpha <1$. For the entanglement entropy, we can always find $K >0$ and $k >0$ such that
$$
P( \Delta E_S \geq x) \leq K e^{- k x}.
$$
We note that $P_{\alpha}^{\rm GEE}$ and $P_{\alpha, {\rm opt}}^{\rm GEE}$ are valid bounds for any $\alpha$, while $P_{\alpha}^{\rm REE}$ is not a valid bound for $\alpha>1$.

As an illustrative example, we choose the initial state $\ket{\chi(r)}_{AB}  \propto \sum_{i=1}^d \sqrt{r^i} \ket{ii}_{AB}$ with $r=0.86$ and $d=500$ to compare the probability bounds on raising the entanglement entropy. Figure~\ref{Fig:CompareBounds} shows a significant gap between the bounds $P_{\alpha}^{\rm REE}$ and $P_{\alpha, {\rm opt}}^{\rm GEE}$ when $\Delta E_S$ becomes large. Even though evaluation of the optimal probability is computationally challenging, we can consider some SLOCC protocols that raise entanglement with nonzero success probabilities. A well-known example is the distillation of the $k$-level maximally entangled state $\ket{\Phi_k}_{AB}  \propto \sum_{i=1}^k \ket{ii}_{AB}$ \cite{Lo01}. We also introduce a more efficient protocol by iteratively mixing the maximum and minimum Schmidt coefficients, $\lambda_1^\downarrow$ and $\lambda_d^\downarrow$, which provides a higher success probability than the $\Phi_k$-distillation protocol (for details of the protocol we make use of, see the Supplemental Material \cite{Suppl}). As mixing the Schmidt coefficients always increases entanglement, such protocol always gives the nonzero success probabilities for any pairs of the initial state and target entanglement, unless $E_{\rm target} > \log d$. Figure~\ref{Fig:CompareBounds} also compares $P_{\alpha}^{\rm REE}$ and $P_{\alpha, {\rm opt}}^{\rm GEE}$ for randomly generated initial states and target entanglement entropies.
\begin{figure}[t]
\includegraphics[width=0.9\linewidth]{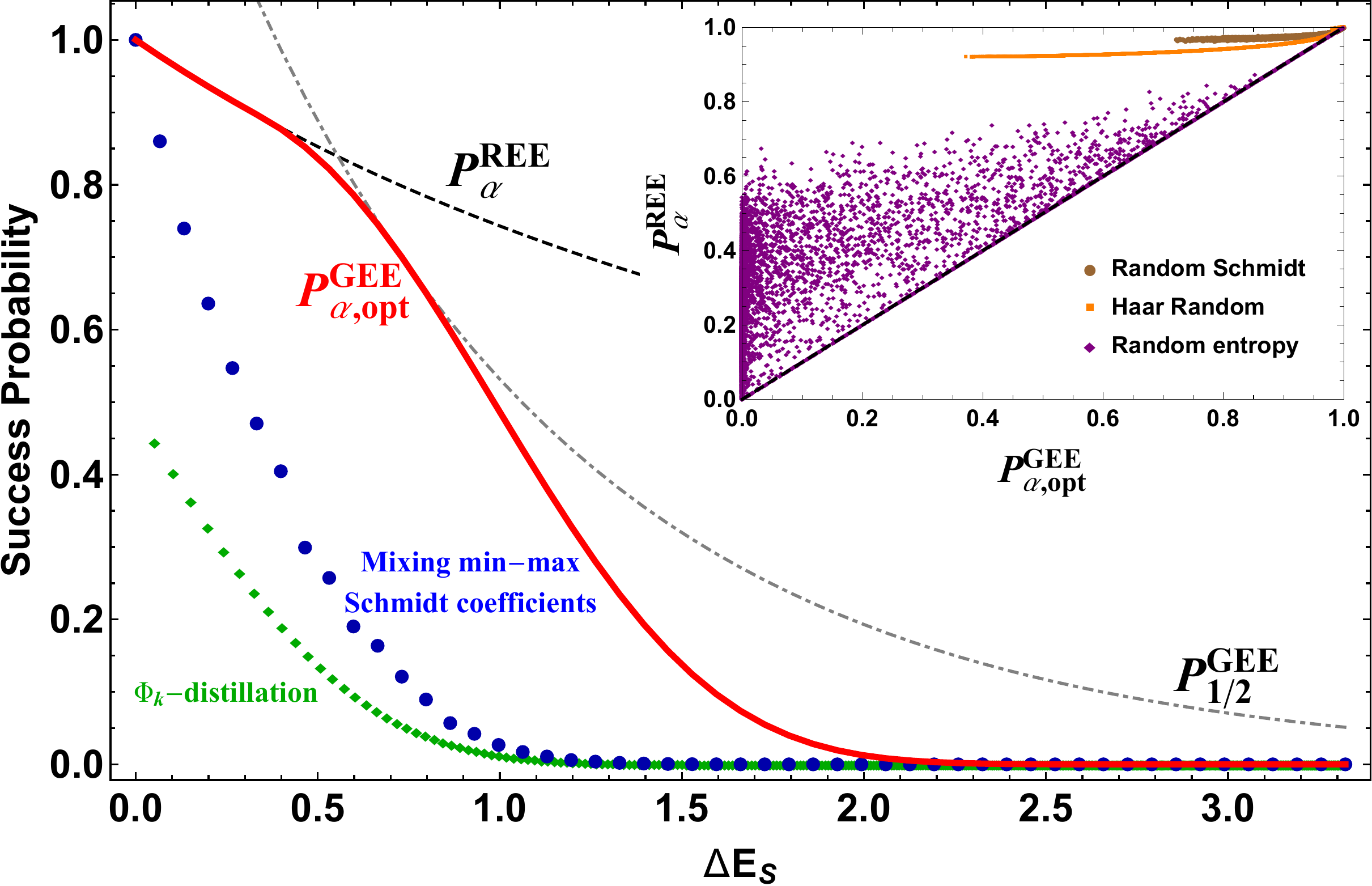}
\caption{Probability bounds given by $P_\alpha^{REE}$ (dashed line), $P_{1/2}^{GEE}$ (dot-dashed line), and $P_{\alpha, {\rm opt}}^{GEE}$ (solid line) for raising the entanglement entropy $(\alpha = 1)$ of the initial state $\ket{\chi(r)}$. Green diamond points represent the distillation probabilities of $\ket{\Phi_k}_{AB}$, while blue circular points represent the success probabilities of the protocol mixing the maximum and minimum Schmidt coefficients. Inset: Comparison between $P_\alpha^{REE}$ and $P_{\alpha, {\rm opt}}^{GEE}$ for randomly generated samples of initial states ($d=500$) and $E_{\rm target}$ \cite{note2}.}
\label{Fig:CompareBounds}
\end{figure}

The probability bound on raising entanglement can also be derived when SLOCC is performed to multiple copies of the initial state, i.e., $\ket{\Psi}^{\otimes n}_{AB} \xrightarrow{{\cal E}_{\rm SLOCC}} \{ p_m, \ket{\Psi_m}_{A' B'} \}$, where $A'$ and  $B'$ are subsystems of $A^{\otimes n}$ and  $B^{\otimes n}$, respectively. By defining the raised amount of the REE per copy as $\Delta \epsilon_\alpha(\Psi_m) = 
[E_\alpha (\Psi_m) - E_\alpha (\Psi^{\otimes n})] /n$, we obtain $P(\Delta \epsilon_\alpha \geq x) \leq e^{ - n \left( \frac{1-\alpha}{\alpha} \right) x}$  for $0<\alpha <1$, implying that the probability bound exponentially decreases as the number of copies increases.

{\it Estimating entanglement from the REE distribution via SLOCC}.--- Until now, we have seen that the lower order REE restricts the success probabilities for raising the higher order REE via SLOCC. Conversely, we show that the distribution of the higher order REE can provide a new method to estimate the lower order REE after applying local measurements. From Theorem~\ref{Thm:REE}, we obtain the following inequality:
\begin{proposition}
\label{Prop:EntBound}
For $ 0 < \alpha \leq \gamma$, the REE under any SLOCC transformation satisfies
\begin{equation}
\label{Eq:EntBound}
E_\alpha(\Psi) \geq  \left( \frac{\alpha}{1-\alpha} \right) \log \langle e^{ \left( \frac{1-\alpha}{\alpha} \right) E_\gamma } \rangle.
\end{equation}
\end{proposition}
We apply our bounds on the REE to estimate entanglement from the distribution of experimentally observable quantities after applying SLOCC. We focus on the case of $\gamma >1$, especially $\gamma = 2$, where its detection requires less resources than those for $\gamma < 1$ \cite{Ekert02, Calabrese04}. In particular, $E_2$ can be detected from a single copy of a quantum state by using cross-correlations between randomised local measurements \cite{Brydges19}. While the direct measurement of $E_2(\Psi)$ lower bounds the REEs of order $\alpha \leq 2$ as $E_\alpha(\Psi) \geq E_2(\Psi)$, further information about $E_\alpha(\Psi)$ can be extracted from the outcome distribution of entanglement $\{ p_m,  E_2(\Psi_m) \}$ after performing a one-way SLOCC, consisting of local positive-operator valued measurements (POVMs). Based on Proposition~\ref{Prop:EntBound}, we obtain $E_\alpha (\Psi) \geq \hat{E}_\alpha (\{ p_m, E_2(\Psi_m) \}) := \left( \frac{\alpha}{1-\alpha} \right) \log  \big[ \sum_m p_m e^{ \left( \frac{1-\alpha}{\alpha} \right) E_2 (\Psi_{m}) } \big]$. For instance, a lower bound on $E_{1/2} (\Psi)$, equivalent to the logarithmic negativity \cite{Plenio05} and the maximum overlap between the maximally entangled state \cite{Konig09} for pure states, can be obtained as $E_{1/2} (\Psi) \geq \log \left\langle  e^{E_2} \right\rangle$. Such entanglement measure has been recognised as a useful correlation quantifier of quantum many-body systems, especially in quantum quenching dynamics \cite{Gray18,  Alba19, Gruber20}. While the direct detection of $E_{1/2} (\Psi)$ remains challenging, its lower bound can be efficiently estimated by our protocol. By optimising over all possible POVMs on one of the subsystems, $\hat{E}_\alpha$ bounds $E_\alpha(\Psi)$ always tighter than the direct measurement of $E_2(\Psi)$, i.e., $E_\alpha (\Psi) \geq \hat{E}_\alpha \geq E_2(\Psi)$ for all $0 < \alpha \leq 2$. We note that the required number of experimental runs to achieve this advantage is significantly smaller than that to obtain $E_\alpha(\Psi)$ from a tomographic reconstruction of the quantum state, without making additional assumptions \cite{Note3}. When $\alpha$ approaches zero $\displaystyle \lim_{\alpha \rightarrow 0} \hat{E}_\alpha(\{ p_m,  E_2(\Psi_m) \}) = \max\{ E_2(\Psi_m) \}$, and the gap between $\hat{E}_\alpha$ and $E_2(\Psi)$ monotonically decreases to zero as $\alpha$ approaches $2$.

\begin{figure}[t]
\centering
\includegraphics[width=.9\linewidth]{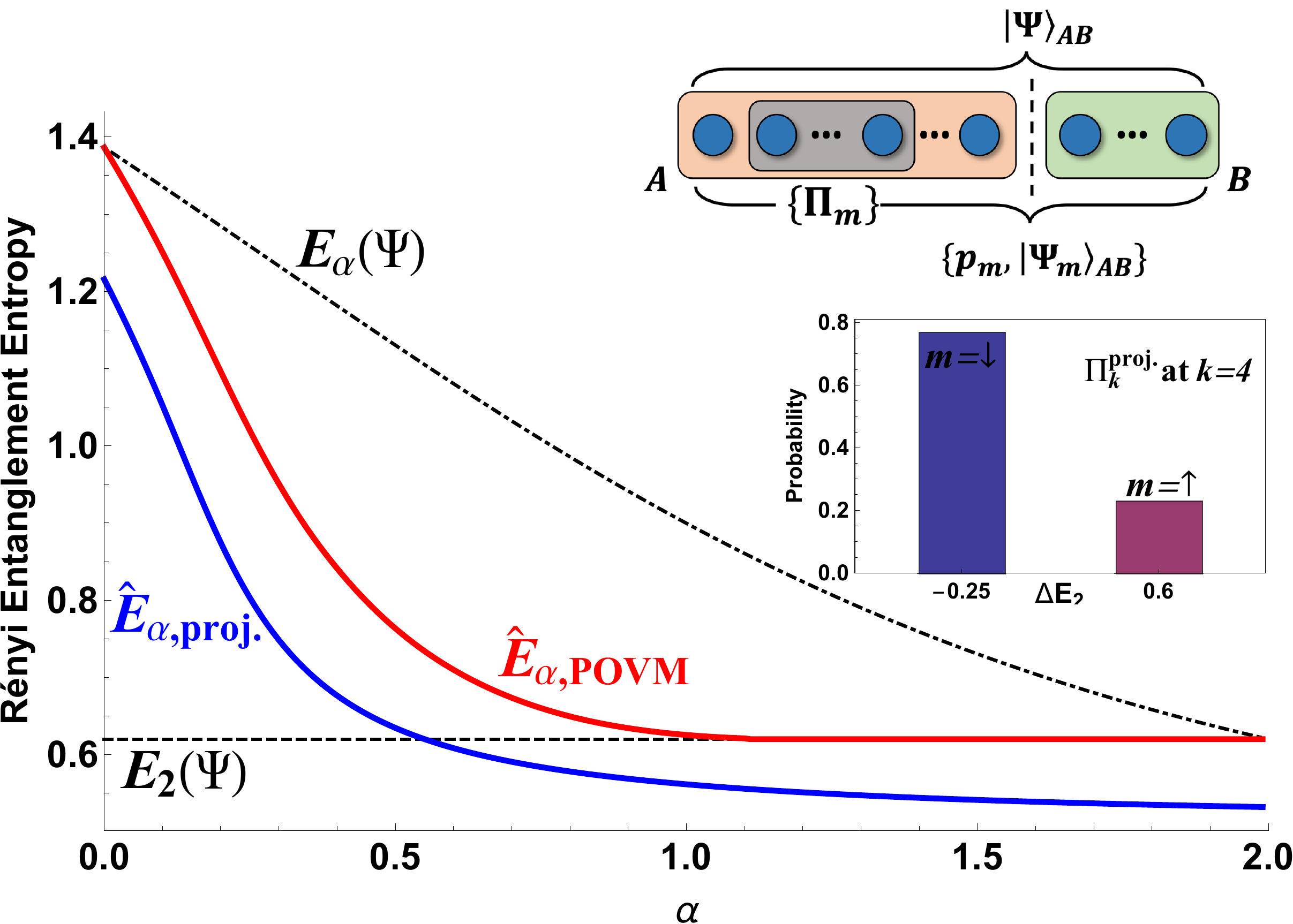}
\caption{Estimation of the REE for the Neel state evolved under the Heisenberg hamiltonian. $\hat{E}_{\alpha, {\rm POVM}}$  (solid-red line) is obtained by optimising over dichotomic POVMs on the subsystem. $\hat{E}_{\alpha, {\rm proj.}}$ (solid-blue line) is evaluated from the outcome REE distribution (inset) after projection measurements on the spin at $k=4$. $E_\alpha(\Psi)$ (dot-dashed line) and $E_2(\Psi)$ (dashed line) are obtained from the exact diagonalisation.}
\label{Fig:Spin}
\end{figure}

As a physical example, we consider the Heisenberg model in an $N$-spin system with the hamiltonian $H = - J \sum_{j=1}^{N} {\vec \sigma}^{(j)} \cdot \vec \sigma^{(j+1)}$ with a periodic boundary condition. Here, $J$ is the interaction strength and $\vec\sigma^{(j)} = (\sigma_x^{(j)}, \sigma_y^{(j)}, \sigma_z^{(j)})$ is the vector of Pauli matrices acting on the $j$th spin. We suppose that the system is initially prepared in the Neel state $\left | \downarrow \uparrow \downarrow \cdots \uparrow \right\rangle$, where  $\left | \uparrow \right \rangle$ and $\left | \downarrow \right \rangle$ are the eigenstates of $\sigma_z$ with the eigenvalues $+1$ and $-1$, respectively. After the state evolves to $\ket{\Psi} = e^{- i H \tau /\hbar} \left | \downarrow \uparrow \downarrow \cdots \uparrow \right\rangle$, we investigate entanglement between two parties with bipartition $N_A$ and $N_B = N- N_A$. Figure~\ref{Fig:Spin} illustrates the REE estimation protocol (POVMs can be performed either on $A$ or $B$) and presents the performance for a specific choice of $\tau = 4.1$, $N=8$ and $N_A=6$. In this case, POVMs on the subsystem $B$ can be realised as collective measurements on two spins as $N_B = 2$. By optimising over dichotomic POVMs on the subsystem, the improvement of the REE estimation for $\alpha =1/2$ is evaluated as $|\hat{E}_{1/2} - E_2(\Psi)| = 0.14$ for the given model and parameters. The protocol can be further simplified by performing single-spin projection measurements $\Pi_k^{\rm proj.} = \{ \left | \uparrow \right\rangle_k \left \langle \uparrow \right |, \left | \downarrow \right\rangle_k \left \langle \downarrow \right | \}$ on the $k$th site, which has been studied in the context of many-body measurement quench dynamics \cite{Bayat18}. While being easier to be implemented, the REE estimation with a single-spin measurement is effective for a limited range of $\alpha$ depending on the physical models and parameters. For our model with $\tau=4.1$, the regime where $\hat{E}_\alpha$ has a better performance than the direct measurement of $E_2(\Psi)$ is numerically verified as $0 \leq \alpha \leq 0.55$.

Our estimation protocol can also be applied for the case when the pure quantum state $\ket{\Psi}_{AB}$  is subject to decoherence and becomes a mixed state $\tilde\rho_{AB} = (1- z) \ket{\Psi}_{AB} \bra{\Psi} +  z \sigma_{AB}$. When $\sigma_{AB}$ is known, e.g., for global depolarisation $\sigma_{AB} \propto {\mathbb 1}_{AB}$, the REE of the uncontaminated state $E_\alpha (\Psi) $ can be estimated from the distribution of the R\'enyi-2 entropy of the subsystem after applying POVMs, similarly to the pure state case. An additional error term due to the decoherence can be calculated based on the continuity of the R\'enyi entropy \cite{Hanson17}. Further discussions on the effect of decoherence along with another physical example, the transverse Ising model in the thermodynamic limits \cite{Vidal03, Latorre04} can be found in the Supplemental Material \cite{Suppl}.

We also note that the right-hand side of Eq.~\eqref{Eq:EntBound} for $\gamma = \alpha$ has been recently studied as a measure for accessible entanglement of indistinguishable particles \cite{Barghathi18} by considering a projection onto the Hilbert space with a fixed particle number in the subsystem.

{\it Generalisation for mixed states}.---We consider generalisation of our results to a mixed state. The GEE for a bipartite mixed state $\rho$ can be constructed as
\begin{equation}
{\rm co}{E}_{(\alpha,s)} (\rho) := \min_{ \{ q_\mu,\psi_\mu \} } \sum_\mu q_\mu E_{(\alpha,s)} (\psi_\mu),
\end{equation}
based on the convex roof construction $\displaystyle {\rm co} f(\rho) := \min_{\{q_\mu, \psi_\mu\}} \sum_\mu q_\mu f(\psi_\mu)$, where $\rho = \sum_\mu q_\mu \ket{\psi_\mu}\bra{\psi_\mu}$. When $s \rightarrow 0$, ${\rm co}E_{(\alpha, s)}$ becomes ${\rm co}E_\alpha$, where its evaluation has been studied \cite{Kim10, Wang16} for some classes of mixed states. For the case of $\alpha \rightarrow 1$, ${\rm co}E_{(\alpha,1/\alpha)}$ becomes the entanglement of formation \cite{Bennett96}.

Meanwhile, a general SLOCC protocol including transforms between mixed bipartite states can be described as a LOCC instrument \cite{Chitambar14}, in which a coarse-grained outcome consists of fine-grained SLOCC outcomes describing transformations between pure states. The following then holds for mixed states:
\begin{proposition} \label{Prop:Mixed}
Suppose that an initial bipartite mixed state $\rho$ transforms by SLOCC. Then, the following inequality holds:
$$
{\rm co}{E}_{(\alpha,s)} (\rho) \geq \frac{1}{s(1-\alpha)}\left[ \langle  e^{ s(1-\alpha) {\rm co}{E}_\gamma} \rangle -1\right]
$$
for $0 < \alpha <1$, $\alpha \leq \gamma$, and $s \leq 1/\alpha$. 
From this, the success probability of raising ${\rm co}E_\alpha$ is upper bounded as 
$$
\begin{aligned}
P\left(  {\rm co}E_\alpha  \geq E_{\rm target} \right)
&\leq \min_{0 \leq  \beta \leq \alpha} \left[ \frac{ \left(\frac{1-\beta}{\beta} \right) {\rm co}{E}_{(\beta,1/\beta)} (\rho) }{e^{\left(\frac{1-\beta}{\beta} \right) E_{\rm target}}-1} \right].
\end{aligned}
$$
\end{proposition}
\noindent Although evaluating ${\rm co}E_{(\alpha,1/\alpha)}$ for a general mixed state is an open question, these bounds can be utlised to study the REEs of the Werner and isotropic states as ${\rm co}E_{(1/2,2)}$ of those states have a less complicated form \cite{Lee03} than the entanglement of formation \cite{Bennett96}.

{\it Remarks}.--- We have shown that under any SLOCC process, there exist refined conditions on the outcome distribution of the REE, beyond its mean value. To this end, we have introduced a new family of entanglement measures based on the generalised entropy, the monotonicity of which involves the contribution of the higher-order moments of the outcome REE distribution. Our work provides a fundamental limitation for stochastic entanglement distillation, namely that its success probability exponentially decreases as the distilled amount of entanglement increases. The refined condition can also be utilised to obtain a lower bound on the initial state's $E_\alpha$ from the distribution of the outcome $E_\gamma$ with $\gamma \geq \alpha$, for instance $E_2$ which can more readily be measured in experiments.

An interesting direction for future research is applying our results to other entanglement quantifiers related to the REE, such as the conditional R\'enyi entropy \cite{Muller13} and $\alpha$-logarithmic negativity \cite{Wang19}. Generalisation of our work to nondeterministic manipulation of multi-partite entanglement \cite{Sauerwein18} could also be an intriguing topic as there exist distinct classes of entangled states that are not interconvertible by SLOCC \cite{Dur00}. 

\begin{acknowledgments}
The authors thank Soojoon Lee, Seok Hyung Lie, and Hongzheng Zhao for helpful discussions. This work is supported by the KIST Open Research Program, the QuantERA ERA-NET within the EU Horizon 2020 Programme, the EPSRC (EP/R044082/1) and the EPSRC Centre for Doctoral Training in Controlled Quantum Dynamics. M. S. K. acknowledges the Royal Society, KIAS visiting professorship and Samsung GRO and GRP grants.
\end{acknowledgments}

\newpage
\newpage
\widetext
\section{Supplemental Material}

\section{I. Proof of Proposition 1} 
\begin{proposition} 
\label{Prop:AntiNorm}
The generalised entanglement entropy (GEE) defined as
\begin{equation}
E_{(\alpha, s)} (\Psi) := \frac{1}{s(1-\alpha)} \left[ e^{ s(1-\alpha) E_\alpha (\Psi) } -1 \right]
\end{equation}
is an entanglement monotone for $(\alpha, s) \in \Omega := \{ (\alpha, s) | 0 < \alpha < 1~{\rm and}~  s \leq 1/\alpha \} \cup \{ (\alpha, s) | \alpha > 1~{\rm and}~ s \geq 1/\alpha \}$, satisfying the following conditions:\\
(i) $E_{(\alpha,s)} (\Psi) = 0$ \textit{if and only if} $\ket{\Psi}_{AB}$ is separable,\\
(ii) $\langle \Delta E_{(\alpha,s)} \rangle \leq  0$ via stochastic local operations and classical communication (SLOCC).\\
When $s \rightarrow 0$, $E_{(\alpha,s)}$ becomes $E_\alpha$, whereas $s = 1$ gives the entanglement measure based on the Tsallis entropy \cite{Tsallis88}.
\end{proposition}
\begin{proof}
As every entanglement monotone for a pure state $\ket{\Psi}_{AB}$ is defined by a concave function on the local density matrix $\rho_B = {\rm Tr}_A \ket{\Psi}_{AB}\bra{\Psi}$ \cite{Vidal00}, it sufficient to show that 
$$
S_{(\alpha, s)} (\rho_B) :=  \frac{1}{s(1-\alpha)} \left[ e^{s(1-\alpha) S_\alpha (\rho_B) } -1 \right],
$$
is an operator concave function for positive semidefinite matrices. This has been proven in Ref.~\cite{Hu06} based on Minkowski's inequality, and we show a simplified version of the proof. We first show the following lemma, regarding 
the Schatten norm $\| X \|_\alpha := \left[ {\rm Tr} (\sqrt{X^\dagger X})^\alpha \right]^{1/\alpha}$.
\begin{lemma} [Concavity (convexity) of $\| X\|_\alpha $ \cite{Hu06, Bourin11}] For positive semidefinite matrices $X$ and $Y$ and $0 \leq \mu \leq 1$,
$$
\| \mu X + (1-\mu) Y \|_\alpha 
\begin{cases}
\geq \mu \| X \|_\alpha + (1-\mu) \| Y \|_\alpha & (0 < \alpha \leq 1)\\
\leq \mu \| X \|_\alpha + (1-\mu) \| Y \|_\alpha & (\alpha \geq 1)
\end{cases}.
$$
\end{lemma}
\begin{proof}
Let us define $\tilde X = X / \| X \|_p$ and $\tilde Y = Y / \| Y \|_p$. Then we have 
$$
\frac{\| \mu X  + (1-\mu) Y \|_\alpha}{\mu \| X \|_\alpha  + (1-\mu) \| Y \|_\alpha }=  \left[ \Tr (\tilde\mu \tilde X  + (1-\tilde\mu) \tilde Y)^\alpha \right]^{1/\alpha},
$$
where $\tilde \mu = \mu \| X\|_\alpha / (\mu \| X \|_\alpha + (1-\mu) \| Y\|_\alpha)$. By noting that $f(t) = t^p$ is concave for $0 < p <1$ and convex for $ p>1$ and by using $\Tr [\tilde X^\alpha ] = 1= \Tr [\tilde Y^\alpha]$, we verify that the right-hand-side of the equation is greater or equal to $1$ for $0 < \alpha <1$ while less or equal to $1$ for $\alpha > 1$.
\end{proof}
Now we show the main proof by noting that 
$S_{(\alpha, s)} (\rho_B) = \frac{1}{s(1-\alpha)} \left[ \| \rho_B \|_\alpha^{\alpha s} -1 \right]$. For $0 < \alpha <1$ and $s \leq 1/\alpha$, $\| \rho_B \|_\alpha^{\alpha s}$ is concave since $\| \rho_B \|_\alpha$ is a monotone increasing function and concave for $0 < \alpha <1$ and $f(t) = t^{\alpha s}$ is a concave function for $\alpha s \leq 1$. Conversely, for $\alpha >1$ and $s \geq 1/\alpha$, $\| \rho_B \|_\alpha^{\alpha s}$ is convex since $\| \rho_B \|_\alpha$ is convex for $\alpha >1$ and $f(t) = t^{\alpha s}$ is a convex function for $\alpha s \geq 1$. By taking into account the term $(1-\alpha)$, which is positive (negative) for $0< \alpha <1$ ($\alpha >1$), we conclude that $S_{(\alpha, s)} (\rho_B)$ is concave on positive semidefinite matrices for $(\alpha, s) \in \Omega$.

We derive the limiting cases of GEEs for $s\rightarrow 0$ and $s=1$ by noting that
$$
\lim_{s\rightarrow 0} S_{(\alpha,s)}(\rho_B) = S_\alpha (\rho_B)
$$
and that the Tsallis entropy is defined as $T_\alpha (\rho_B) :=  \frac{1}{1-\alpha} \left( \Tr\rho^\alpha_B -1 \right) = \frac{1}{1-\alpha} \left[ e^{(1-\alpha) S_\alpha (\rho_B)} -1 \right] = S_{(\alpha,1)}$. In particular when $\alpha$ approaches $1$, $E_{(\alpha, s)}$ becomes the entanglement entropy $E_S$ for any $s \neq 0$ as
$$
\lim_{\alpha \rightarrow 1} \frac{1}{s(1-\alpha)} \left[ e^{s(1-\alpha) S_\alpha (\rho_B) } -1 \right]= S(\rho_B) = E_S(\Psi).
$$
\end{proof}

\newpage Based on the monotonicity of the GEE we derive our main theorem, a refined statistical property that the REEs obey under SLOCC.
\section{II. Proof of Theorem 1}
We start with the following inequality given by the monotonicity of $E_{(\alpha, s)} $,
$$
\frac{1}{s(1-\alpha)} \left[ \sum_m p_m e^{ s (1-\alpha) E_{\alpha} (\Psi_m)} -1\right] \leq \frac{1}{s(1-\alpha)} \left[ e^{ s( 1-\alpha )E_{\alpha} (\Psi)} -1 \right],
$$
which can be rearranged to
$$
\frac{1}{s(1-\alpha)} \langle e^{ s (1-\alpha) \Delta E_{\alpha}}  \rangle \leq \frac{1}{s(1-\alpha)} .
$$
We then have the desired inequality, depending on the sign of $(1-\alpha)$ ,
\begin{equation}
\label{Eq:Appd1}
\langle e^{ s (1-\alpha) \Delta E_{\alpha}} \rangle 
\begin{cases}
  \leq 1 & (0 < \alpha < 1~{\rm and}~  s \leq \frac{1}{\alpha} ) \\
  \geq 1 & ( \alpha > 1~{\rm and}~ s \geq \frac{1}{\alpha})
\end{cases}.
\end{equation}

\section{III. Optimal success probability of raising entanglement via SLOCC}
We show that for an entanglement monotone $E$, the optimal success probability of raising entanglement more than a desired value $E_{\rm target}$ via SLOCC is given by
\begin{equation}
\sup_{{\cal E}_{\rm SLOCC}} P \left(  E  \geq E_{\rm target} \right) = \max_{\Psi'} \left[ \min_{l \in \{ 1,2, \cdots, d\} } \left( \frac{\sum_{i=l}^d \lambda_i^\downarrow(\Psi)} {\sum_{i=l}^d \lambda_i^\downarrow(\Psi')} \right) \vline E(\Psi') = E_{\rm target} \right],
\end{equation}
where $\lambda_i^\downarrow(\Psi)$ are the Schmidt coefficients of $\ket{\Psi}_{AB}$ in descending order. 
\begin{proof}
We first show that a single state transformation is enough to achieve the optimal probability. In order to prove this, we use that the necessary and sufficient condition  \cite{Jonathan99} for a pure state transition under an SLOCC protocol:
\begin{equation}
\label{Eq:LOCCcond}
\ket{\Psi}_{AB} \xrightarrow{{\cal E}_{\rm SLOCC}} \{ p_m, \ket{\Psi_m}_{AB} \}	\Longleftrightarrow
 \sum_{i=l}^d \lambda^\downarrow_{i} (\Psi) \geq \sum_m p_m \left( \sum_{i=l}^d \lambda_{i}^\downarrow(\Psi_m) \right)~\forall l = 1, 2, \cdots, d.
\end{equation}

Let us suppose that there exists an SLOCC protocol that transforms the initial bipartite state $\ket{\Psi}_{AB} = \sum_{i=1}^d \sqrt{\lambda_i^\downarrow(\Psi)} \ket{ii}_{AB}$ into $ \{ p_m, \ket{\Psi_m}_{AB} \}$, where $\ket{\Psi_1}_{AB} = \sum_{i=1}^d \sqrt{\lambda_i^\downarrow{(\Psi_1)}} \ket{ii}_{AB}$ and $\ket{\Psi_2}_{AB} = \sum_{i=1}^d \sqrt{\lambda_i^\downarrow{(\Psi_2)}} \ket{ii}_{AB}$ satisfy $E(\Psi_1) \geq E_{\rm target}$ and $E(\Psi_2) \geq E_{\rm target}$. Without loss of generality, we assume that all states have the same Schmidt basis since this could be achieved via local unitary operations.
We then define the following bipartite state from $\ket{ \Psi_1}_{AB}$ and $\ket{\Psi_2}_{AB}$:
$$
\ket{\Psi'}_{AB} = \sum_{i=1}^d \sqrt{ \left( \frac{p_1}{p_1+p_2} \right) \lambda_i^\downarrow{(\Psi_1) + \left( \frac{p_2}{p_1+p_2} \right) \lambda_i^\downarrow{(\Psi_2) }}} \ket{ii}_{AB} = \sum_{i=1}^d \sqrt{\lambda_i^\downarrow(\Psi')} \ket{ii}_{AB},
$$
where $\lambda_i^\downarrow(\Psi')$ is also given in descending order.
We note that any entanglement monotone $E$ of a pure bipartite state $\ket{\Psi}_{AB}$ can be expressed by using a concave function $V$ on its (positive-semidefinite) local density matrix $\rho_{B}$ as $E(\Psi) = V(\rho_B)$ \cite{Vidal00}. We than have 
$$
\begin{aligned}
E(\Psi') &= V \left(  \left( \frac{p_1}{p_1+p_2} \right) \rho_{B1} + \left( \frac{p_2}{p_1+p_2} \right)  \rho_{B2} \right)\\
& \geq  \left( \frac{p_1}{p_1+p_2} \right) V(\rho_{B1})+ \left( \frac{p_2}{p_1+p_2} \right) V(\rho_{B2})\\
& =  \left( \frac{p_1}{p_1+p_2} \right) E(\Psi_1)+ \left( \frac{p_2}{p_1+p_2} \right) E(\Psi_2)\\
&\geq E_{\rm target},
\end{aligned}
$$
where the first inequality comes from the concavity of $V$ on a set of positive-semidefinite matrices. Here, $\rho_{B1}$ and  $\rho_{B2}$ are reduced density matrices of $\ket{\Psi_1}_{AB}$ and $\ket{\Psi_2}_{AB}$, respectively. Also, from the necessary and sufficient condition for an SLOCC transformation given by Eq.~\eqref{Eq:LOCCcond}, we have
$$
\begin{aligned}
 \sum_{i=l}^d \lambda_i^\downarrow (\Psi) &\geq  p_1 \left( \sum_{i=l}^d \lambda_i^\downarrow(\Psi_1) \right) + p_2 \left( \sum_{i=l}^d \lambda_i^\downarrow(\Psi_2) \right) + \sum_{m \neq 1,2} p_m \left( \sum_{i=l}^d \lambda_{i}^\downarrow(\Psi_m) \right)\\
 &= (p_1 + p_2) \left[ \sum_{i=l}^d \left( \frac{p_1}{p_1+p_2} \right) \lambda_i^\downarrow{(\Psi_1) + \left( \frac{p_2}{p_1+p_2} \right) \lambda_i^\downarrow{(\Psi_2) }} \right] + \sum_{m \neq 1,2} p_m \left( \sum_{i=l}^d \lambda_{i}^\downarrow(\Psi_m) \right)\\
 &= (p_1 + p_2) \left( \sum_{i=l}^d \lambda_i^\downarrow(\Psi') \right)+ \sum_{m \neq 1,2} p_m \left( \sum_{i=l}^d \lambda_{i}^\downarrow(\Psi_m) \right)
\end{aligned}
$$
for all $l = 1, 2, \cdots, d$. Therefore, there exists an SLOCC protocol that transforms $\ket{\Psi}_{AB}$ into $\ket{\Psi'}_{AB}$ with success probability $p_1 + p_2$.

Next, by noting that $E(\Psi') \geq E_{\rm target}$, we can always find a deterministic LOCC protocol ${\cal E}_{\rm LOCC}$ such that $E({\cal E}_{\rm LOCC}(\Psi')) = E_{\rm target}$. By combining these two protocols, we observe that for any SLOCC protocol that transforms $\ket{\Psi}_{AB}$ into two outcome states $\ket{\Psi_1}_{AB}$ and $\ket{\Psi_2}_{AB}$ satisfying $E (\Psi_{1(2)})\geq E_{\rm target}$ with probabilities $p_1$ and $p_2$, we can always find an SLOCC protocol that transforms $\ket{\Psi}_{AB}$ into $\ket{\Psi'}_{AB}$ satisfying $E (\Psi') = E_{\rm target}$ with probabilities $p_1 + p_2$, i.e. with the same accumulated probability. This can be generalised for multiple outcome states of $\ket{\Psi_m}_{AB}$ satisfying $E (\Psi_m) \geq E_{\rm target}$ by applying multiple rounds of the pairwise merging process described above. Therefore, optimisation of the accumulated probability over all possible outcome states can be reduced to optimisation over a transition probability to a single state $\ket{\Psi'}_{AB}$ such that $E (\Psi') = E_{\rm target}$.
\end{proof}

\section{IV. Entanglement manipulation protocol by varying the maximum and minimum Schmidt coefficients}
We introduce a simple protocol that can give a higher success probability than the distillation of the $k$-level maximally entangled state $\ket{\Phi_k}_{AB}  \propto \sum_{i=1}^k \ket{ii}_{AB}$. Suppose that the Schmidt decomposition of the initial bipartite state is given by $\ket{\Psi}_{AB} = \sum_{i=1}^d \sqrt{\lambda_i^\downarrow(\Psi)} \ket{ii}_{AB}$. 
We find the target state $\ket{\Psi'}_{AB}$ by the following steps:

For the first round, we vary the values of the maximum and minimum Schmidt coefficients $\lambda_1^\downarrow (\Psi)$ and $\lambda_d^\downarrow (\Psi) $ to $\lambda_1^\downarrow (\Psi) - \epsilon$ and $\lambda_d^\downarrow (\Psi) + \epsilon$, while keeping the others unchanged. In order to make all the Schmidt coefficients non-negative, $0 \leq \epsilon \leq \min \{ \lambda_1^\downarrow, 1- \lambda_d^\downarrow \}$. Then we have two possible situations:

(i) If there exists $\epsilon$ such that the outcome entanglement reaches $E_{\rm target}$, we update the state to $\ket{\Psi'}_{AB}$ having the same Schmidt coefficients as the initial state $\ket{\Psi}_{AB}$, except the two elements $\lambda_1^\downarrow (\Psi) - \epsilon$ and $\lambda_d^\downarrow (\Psi) + \epsilon$. Then entanglement of the target state becomes $E(\Psi') = E_{\rm target}$, and we finish the protocol.

(ii) If there is no $\epsilon$ that can reach $E_{\rm target}$ from varying $\lambda_1^\downarrow (\Psi) \rightarrow \lambda_1^\downarrow (\Psi) -\epsilon$ and $\lambda_d^\downarrow (\Psi) \rightarrow \lambda_d^\downarrow (\Psi) + \epsilon$, we update both coefficients to $(\lambda_1^\downarrow (\Psi)  + \lambda_d^\downarrow (\Psi) ) /2$. 

After each round of mixing the Schmidt coefficients, the degree of entanglement always increases, i.e., $E(\Psi') \geq E(\Psi)$. We repeat this using the updated state $\ket{\Psi'}_{AB} = \sum_{i=1}^d \sqrt{\lambda_i^\downarrow(\Psi')} \ket{ii}_{AB}$ recursively until the entanglement of the final state reaches $E_{\rm target}$. We note that for the extreme case, $E_{\rm target} = \log d$, the target state $\ket{\Psi'}_{AB}$ ends up with the maximally entangled state after running sufficiently many rounds. 

Finally, by using Eq.~\eqref{Eq:LOCCcond}, we obtain the SLOCC transformation probability from $\ket{\Psi}_{AB}$ to $\ket{\Psi'}_{AB}$, which is nonzero when the Schmidt ranks of $\ket{\Psi}_{AB}$ and $\ket{\Psi'}_{AB}$ are the same. Therefore, the protocol always gives the nonzero success probabilities for any pairs of the initial state and target entanglement, unless $E_{\rm target} > \log d$.

\section{V. Proof of Proposition 2}
We provide the extended version of Proposition 2 and its proof.
\begin{proposition} 
The accumulated success probability of achieving the outcome $E_\alpha$ larger than $E_{\rm target}$ is upper bounded by
\begin{equation}
\label{Eq:PalphaBound}
P\left(  E_\alpha  \geq E_{\rm target}\right) \leq \frac{ e^{ s (1-\alpha) E_\alpha(\Psi) } -1 } { e^{ s( 1-\alpha) E_{\rm target} }-1},
\end{equation}
for all $(\alpha, s) \in \Omega$. Furthermore, the right-hand-side of the inequality is minimised when $s=1/\alpha$.
\end{proposition}
\begin{proof}
We note that 
$$
\begin{aligned}
\langle e^{s(1-\alpha) \Delta E_\alpha} \rangle &= \sum_m p_m e^{ s (1-\alpha) (E_{\alpha} (\Psi_m) - E_{\alpha} (\Psi))} \\
&= e^{ -s (1-\alpha) E_{\alpha} (\Psi)} \left[  \sum_{E_\alpha(\Psi_m)  \geq E_{\rm target}}p_m e^{ s (1-\alpha)  E_{\alpha} (\Psi_m)}  +  \sum_{E_\alpha(\Psi_m) < E_{\rm target}} p_me^{ s (1-\alpha)  E_{\alpha} (\Psi_m)} \right] \\
&\begin{cases}
\geq e^{ -s (1-\alpha) E_{\alpha} (\Psi)} \left[ P(E_\alpha \geq E_{\rm target}) e^{ s (1-\alpha) E_{\alpha} (\Psi_m)} + P(E_\alpha < E_{\rm target}) \right] & (0 < \alpha < 1)\\
\leq e^{ -s (1-\alpha) E_{\alpha} (\Psi)} \left[ P(E_\alpha \geq E_{\rm target}) e^{ s (1-\alpha) E_{\alpha} (\Psi_m)} + P(E_\alpha < E_{\rm target}) \right] & (\alpha > 1)
\end{cases},
\end{aligned}
$$
where the inequality comes from the fact that $E_\alpha(\Psi_m)  \geq E_{\rm target}$ for the first term and $E_\alpha(\Psi_m)  \geq 0$ for the second term depending on the sign of $(1-\alpha)$. Then using $P(E_\alpha < E_{\rm target}) = 1- P(E_\alpha \geq E_{\rm target})$ and combining with the condition given in Eq.~\eqref{Eq:Appd1}, we have 
$$
\begin{cases}
e^{ -s (1-\alpha) E_{\alpha} (\Psi)} \left[ P(E_\alpha \geq E_{\rm target}) (e^{ s (1-\alpha) E_{\alpha} (\Psi_m)} -1 ) + 1\right] \leq \langle e^{s(1-\alpha) \Delta E_\alpha} \rangle  \leq 1 & (0 < \alpha < 1~{\rm and}~  s \leq \frac{1}{\alpha} ) \\
e^{ -s (1-\alpha) E_{\alpha} (\Psi)} \left[ P(E_\alpha \geq E_{\rm target}) (e^{ s (1-\alpha) E_{\alpha} (\Psi_m)} -1 ) + 1\right] \geq \langle e^{s(1-\alpha) \Delta E_\alpha} \rangle  \geq 1 & ( \alpha > 1~{\rm and}~ s \geq \frac{1}{\alpha})
\end{cases}.
$$
Finally be rearranging the inequalities and noting that $e^{ s (1-\alpha) E_{\alpha} (\Psi_m)} -1 \geq 0 $ for $0<\alpha< 1$ and $e^{ s (1-\alpha) E_{\alpha} (\Psi_m)} -1 \leq 0 $ for $\alpha < 1$, we obtained the desired inequality.

Next, we show that $s=1/\alpha$ gives the minimum bound for any $\alpha \in (0, \infty)$. It is enough to consider the case $E_{\rm target} \geq E_\alpha (\Psi)$, otherwise $P(E_\alpha \geq E_{\rm target}) \geq 1$ only gives a trivial bound. Let us define a function
$$
f(s) := \frac{e^{sx} - 1}{e^{sy} -1}
$$
for given values of $x$ and $y$. We then note that for all $s>0$,
$$
\frac{d f(s)}{ds} = \left( \frac{e^{sx}-1}{e^{sy}-1} \right) \left[ \frac{x e^{sx}}{e^{sx}-1} -  \frac{y e^{sy}}{e^{sy}-1} \right]
\begin{cases}
\geq 0 & (0 \geq x \geq y) \\
\leq 0 & (0 \leq x \leq y)
\end{cases},
$$
since $ \frac{x e^{sx}}{e^{sx}-1} $ is a monotonically increasing function on $x \in (-\infty, \infty)$ for  $s>0$. Then by taking $x= (1-\alpha) E_\alpha (\Psi)$ and $y = (1-\alpha)E_{\rm target}$, we can observe that the bound is monotonically decreasing on $0<s \leq 1/\alpha$ when $0< \alpha <1$, thus the minimum value is achieved for $s=1/\alpha$. Conversely, for $\alpha >1$, the bound is monotonically increasing on  $s \geq 1/\alpha$, so the minimum value is again given by $s = 1/\alpha$.
\end{proof}

\section{VI. Probability bounds for $\alpha = 0$ and $\alpha = 1$}
We show the limiting cases of the probability bound
$$
P_\alpha^{\rm GEE} := \frac{ e^{ \left( \frac{1-\alpha}{\alpha} \right) E_\alpha(\Psi) } -1 } { e^{ \left( \frac{1-\alpha}{\alpha} \right) E_{\rm target} }-1}
$$
when $\alpha$ approaches $0$ and $1$. We first consider the case $\alpha \rightarrow 0$. We can rewrite the bound as
$$
\lim_{\alpha \rightarrow 0} P_\alpha^{\rm GEE} = \lim_{\alpha \rightarrow 0} \left[ \frac{e^{\left( \frac{1-\alpha}{\alpha} \right)(E_\alpha (\Psi) - E_{\rm target})} - e^{-\left( \frac{1-\alpha}{\alpha} \right) E_{\rm target}}}{ 1- e^{-\left( \frac{1-\alpha}{\alpha} \right) E_{\rm target}}}\right],
$$
then it is straightforward to see that
$$
\lim_{\alpha \rightarrow 0} P_\alpha^{\rm GEE} =
\begin{cases}
0 & (E_0(\Psi) = \log d < E_{\rm target})\\
\infty & (E_0(\Psi) = \log d > E_{\rm target})
\end{cases}.
$$
For $E_0(\Psi) = \log d = E_{\rm target}$, we note that
$$
\lim_{\alpha \rightarrow 0 } \left( \frac{1-\alpha}{\alpha} \right) \left[ E_\alpha(\Psi) - E_0(\Psi) \right]  = \lim_{\alpha \rightarrow 0 } \left[  \frac{ \frac{\partial (E_\alpha(\Psi) - E_0(\Psi)) }{\partial \alpha}} {\frac{\partial \left( \frac{\alpha}{1-\alpha} \right) }{\partial \alpha}} \right] = \lim_{\alpha \rightarrow 0 } \left[ - S( \tilde \rho_{B,\alpha} \| \rho_B ) \right],
$$
where $S(\sigma \| \rho) = \Tr \sigma (\log \sigma - \log \rho)$ is the relative entropy and $\tilde \rho_{B,\alpha} = \rho_B^\alpha / (\Tr \rho_B^\alpha )$. For $\alpha \rightarrow 0$, $\tilde \rho_{B, \alpha}$  becomes ${\mathbb 1}_B /d$, then $\lim_{\alpha \rightarrow 0 } \left( \frac{1-\alpha}{\alpha} \right) \left[ E_\alpha(\Psi) - E_0(\Psi) \right]  = - S({\mathbb 1}_B /d \| \rho) = \log d + (1/d) \Tr [\log \rho_B]$. Hence, the probability bound becomes 
$$
\lim_{\alpha \rightarrow 0} P_\alpha^{\rm GEE}  =
\begin{cases}
0 & (E_0(\Psi) = \log d < E_{\rm target})\\
 d e^{ (1/d) {\rm Tr} [\log \rho_B]} & (E_0(\Psi) = \log d = E_{\rm target})\\
 \infty & (E_0(\Psi) = \log d > E_{\rm target})
 \end{cases}.
$$

When $\alpha$ approaches $1$, we can rewrite the bound as
$$
\lim_{\alpha \rightarrow 1 } P_\alpha^{\rm GEE} = \lim_{\alpha \rightarrow 1 }\frac{ \left( \frac{\alpha}{1-\alpha} \right) \left[ e^{ \left( \frac{1-\alpha}{\alpha} \right) E_\alpha(\Psi) } -1 \right] } { \left( \frac{\alpha}{1-\alpha} \right) \left[ e^{ \left( \frac{1-\alpha}{\alpha} \right) E_{\rm target} }-1 \right]} = \frac{E_S(\Psi)}{E_{\rm target}},
$$
since $\displaystyle \lim_{\alpha \rightarrow 1 } \left( \frac{\alpha}{1-\alpha} \right) \left[ e^{ \left( \frac{1-\alpha}{\alpha} \right) E_\alpha(\Psi) } -1 \right] = E_S(\Psi)$ and $\displaystyle \lim_{\alpha \rightarrow 1 } \left( \frac{\alpha}{1-\alpha} \right) \left[ e^{ \left( \frac{1-\alpha}{\alpha} \right) E_{\rm target}} -1 \right] = E_{\rm target}$.

\section{VII. Proof of Proposition 3}
We provide the extended version of Proposition 3 and its proof.
\begin{proposition}
\label{Prop:EntBound}
For $(\alpha, s) \in \Omega$ and $ \alpha \leq \gamma$, the REE under any SLOCC transformation satisfies
\begin{equation}
\label{Eq:EntBound}
E_\alpha(\Psi) \geq  \frac{1}{s(1-\alpha)} \log \langle e^{ s( {1-\alpha})  E_\gamma } \rangle,
\end{equation}
where the right-hand-side of the inequality is maximised when $s=1/\alpha$.
\end{proposition}
\begin{proof}
It is straightforward to obtain the inequality by noting that 
$$
 \langle e^{s(1-\alpha) \Delta E_\alpha} \rangle = e^{-s(1-\alpha) E_\alpha (\Psi)} \langle e^{s(1-\alpha)E_\alpha} \rangle
 \begin{cases}
  \leq 1 & (0 < \alpha < 1~{\rm and}~  s \leq \frac{1}{\alpha} ) \\
  \geq 1 & ( \alpha > 1~{\rm and}~ s \geq \frac{1}{\alpha})
\end{cases},
$$
where the inequality comes from Theorem 1. By rearranging the inequality and taking into account the sign of $(1-\alpha)$ we obtain the desired inequality. We note that the right-hand-side of the inequality is a monotone on $E_\alpha (\Psi_m)$, thus we obtain
$$
E_\alpha(\Psi) \geq  \frac{1}{s(1-\alpha)} \log \langle e^{ s( {1-\alpha})  E_\alpha } \rangle \geq \frac{1}{s(1-\alpha)} \log \langle e^{ s( {1-\alpha})  E_\gamma } \rangle
$$
as $E_\gamma \leq E_\alpha$ for $0 < \alpha \leq \gamma$.

We now show that $s=1/\alpha$ gives the maximum values of the bound,
$$
g(s) := \frac{1}{s} \log \left( \sum_m p_m e^{s x_m} \right)
$$
for a probability distribution $\{p_m\}$ with outcome entities $x_m$. We then note that 
$$
\frac{d g(s)}{ds} = \frac{1}{s^2} \left[ \sum_m \left( \frac{ p_m e^{s x_m}}{\sum_{m'} p_{m'} e^{s x_{m'}}} \right) \log \left( \frac{e^{s x_m}}{\sum_{m''} p_{m''} e^{s x_{m''}}} \right) \right] = \frac{ H( \tilde p \| p)}{s^2} \geq 0,
$$
where $H(\tilde p \| p) = \sum_m \tilde p_m \log( \tilde p_m / p_m)$ is the (classical) relative entropy between two probability distributions $\{ \tilde p_m =  p_m e^{s x_m} / (\sum_{m'} p_{m'} e^{s x_{m'}}) \}$ and $\{ p_m \}$. From this result, we note that $\frac{1}{s(1-\alpha)} \log \langle e^{ s( {1-\alpha})  E_\gamma } \rangle$ is a monotonically increasing (decreasing) function of $s$ when $0 < \alpha <1$ ($\alpha >1$). Hence, $s = 1/\alpha$ gives the maximum bound for all $\alpha \in (0, \infty)$ and any given distribution $\{ p_m, E_\gamma (\Psi_m) \}$ after applying SLOCC.
\end{proof}

\section{VIII. Estimating the REE of quantum many-body systems}
\begin{figure}[b]
\includegraphics[width=.33\linewidth]{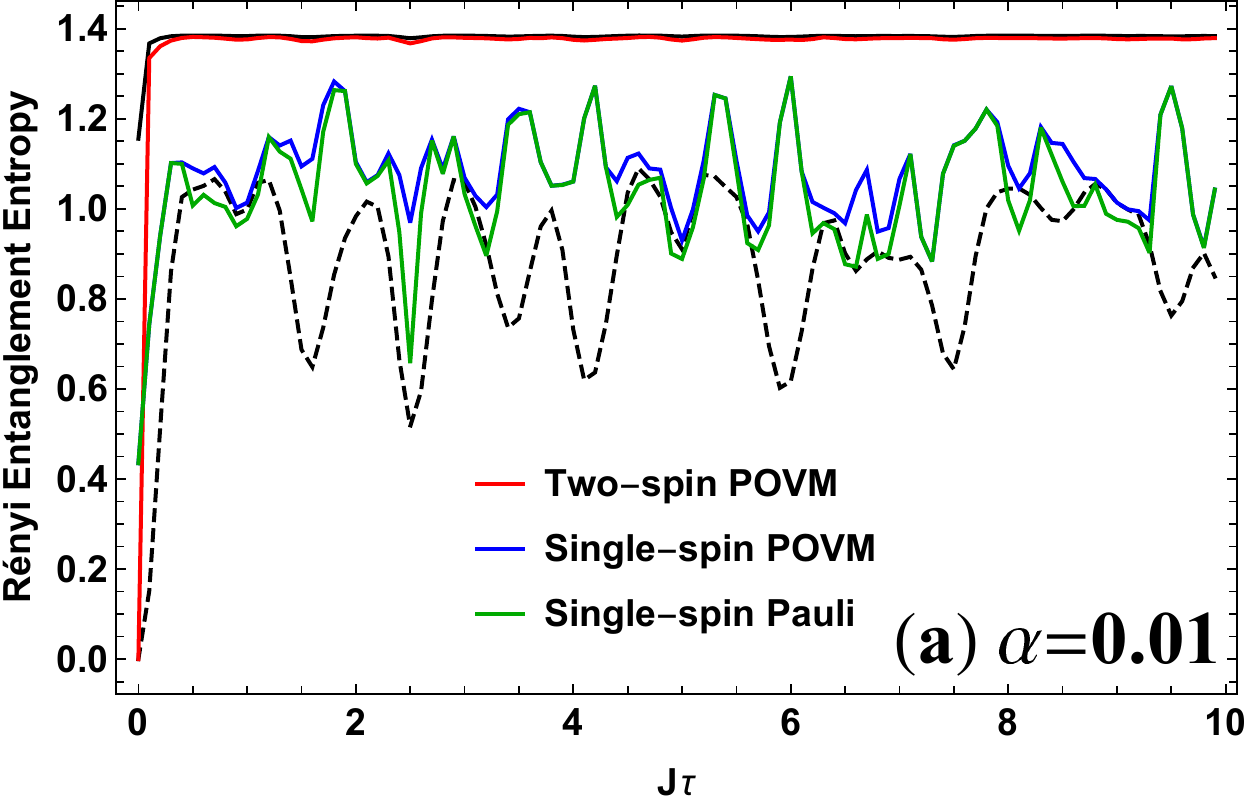}\qquad
\includegraphics[width=.33\linewidth]{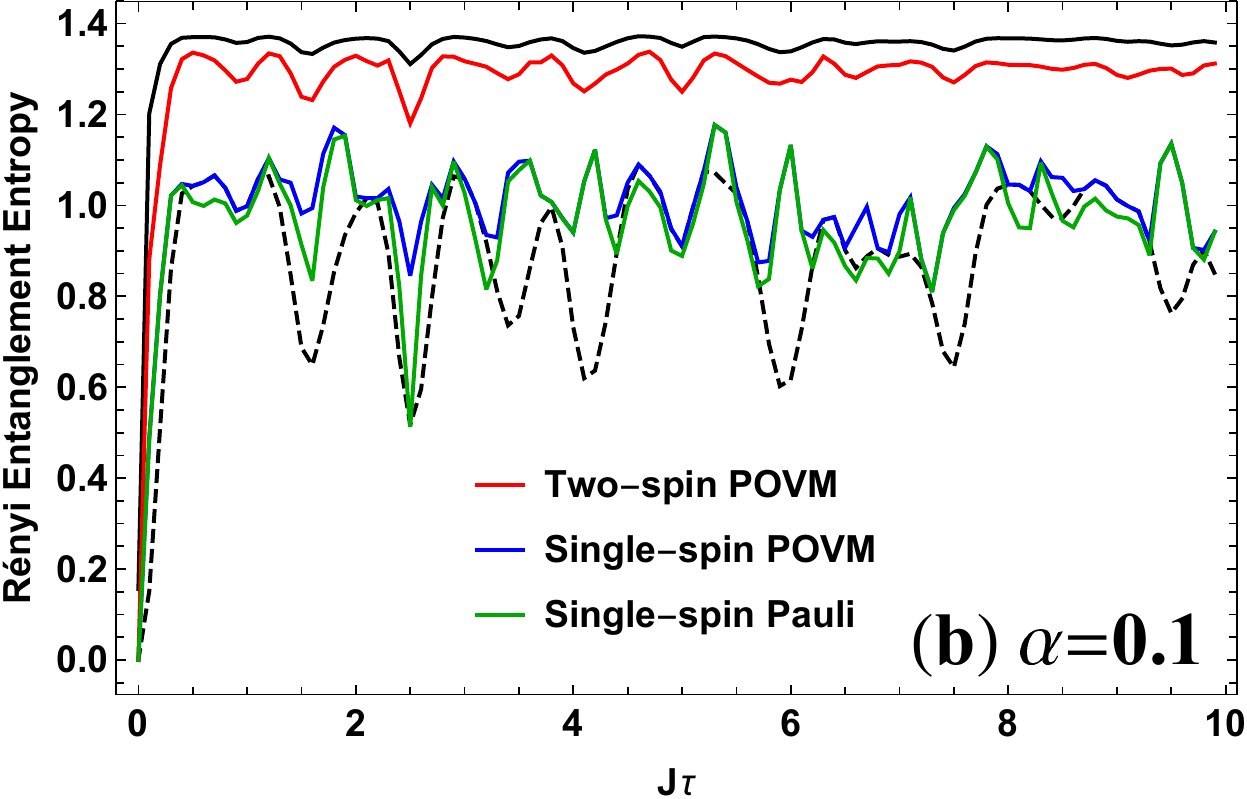}\\
\includegraphics[width=.33\linewidth]{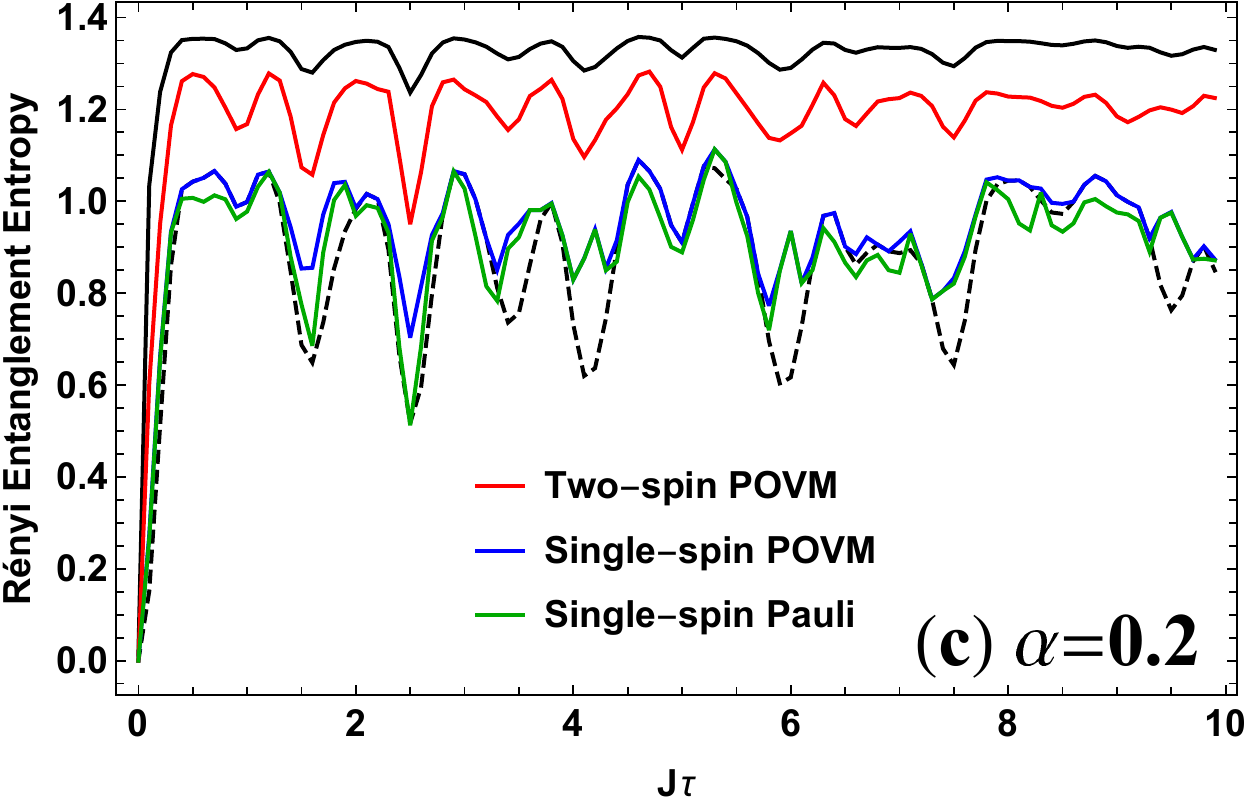}\qquad
\includegraphics[width=.33\linewidth]{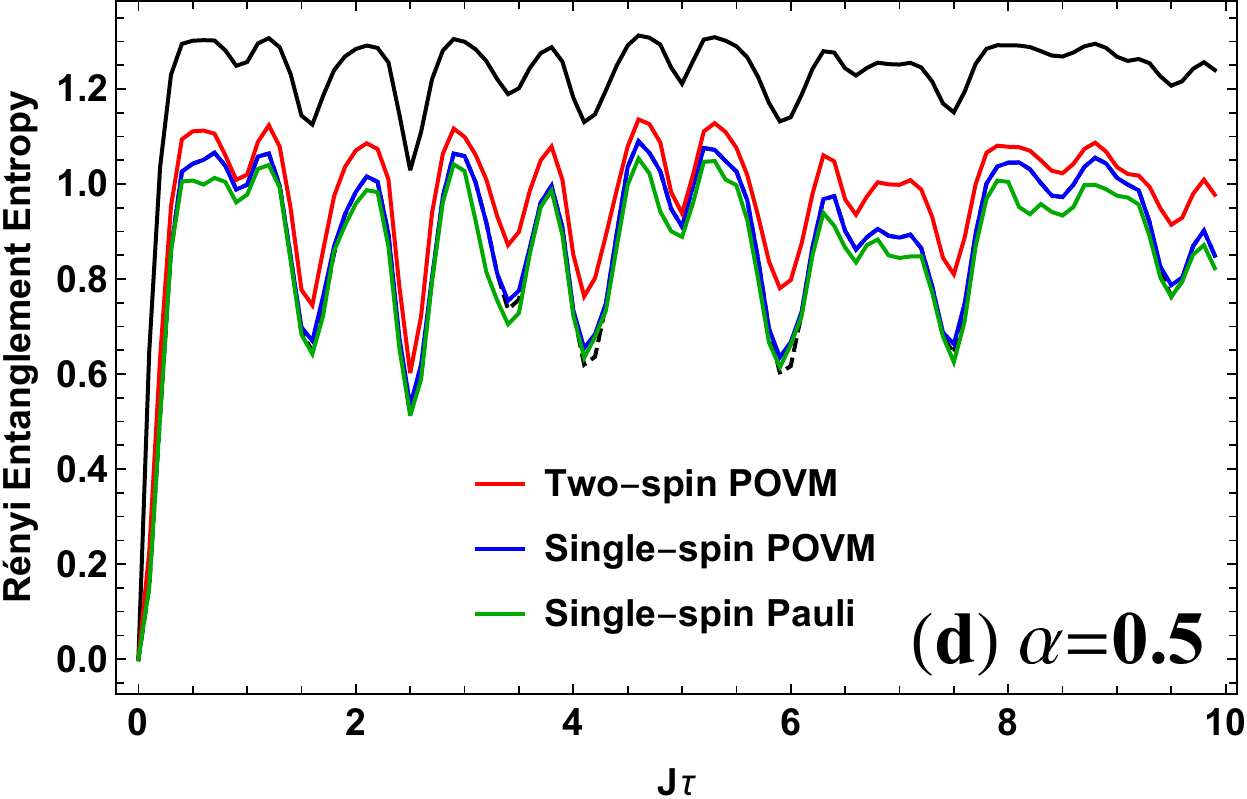}
\caption{Exact values of $E_\alpha(\Psi)$ (black-solid lines) and estimated lower bounds $\hat{E}_\alpha(\Psi)$ for (a) $\alpha = 0.01$, (b) $\alpha = 0.1$, (c) $\alpha = 0.2$, and (d) $\alpha=0.5$ for the Neel state evolving under the Heisenberg Hamiltonian, $\ket\Psi =  e^{- i H \tau /\hbar} \left | \downarrow \uparrow \downarrow \cdots \uparrow \right\rangle$. Dashed lines refer to the exact values of $E_2(\Psi)$.}
\label{Fig:Suppl1}
\end{figure}

\begin{figure}[t]
\includegraphics[width=.32\linewidth]{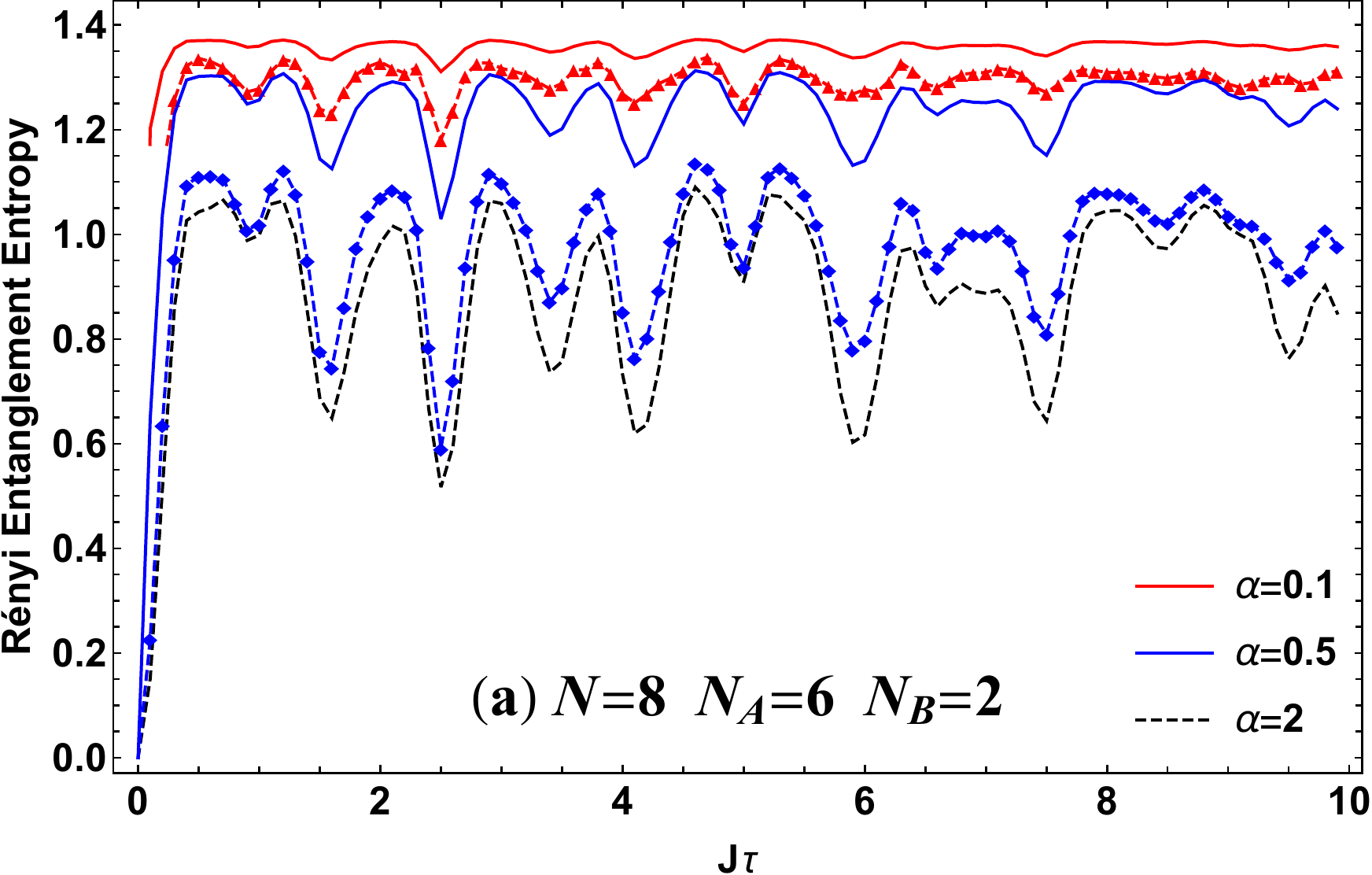}~
\includegraphics[width=.32\linewidth]{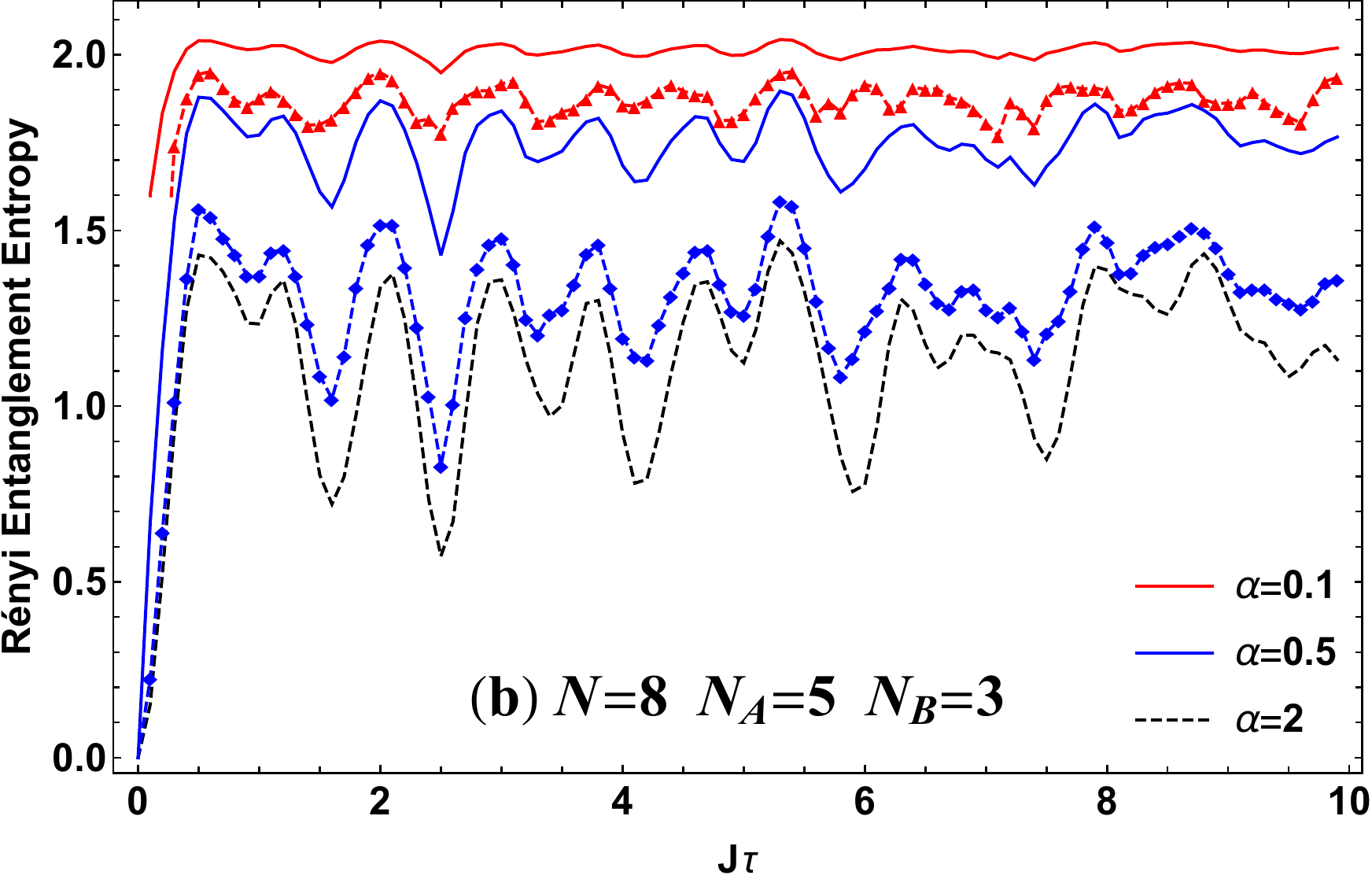}~
\includegraphics[width=.32\linewidth]{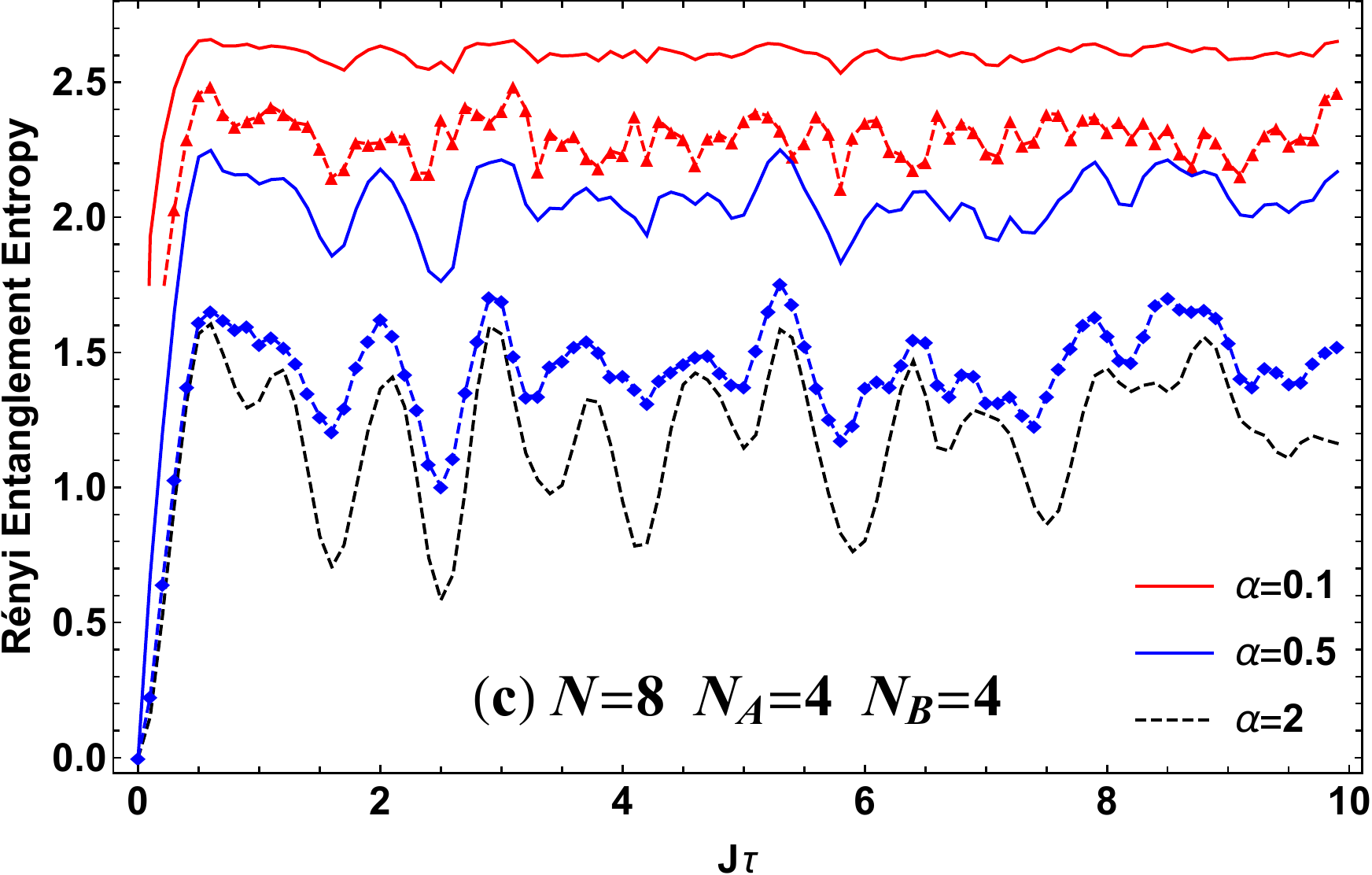}
\caption{Estimation of the REEs with $\hat{E}_{\alpha, {\rm POVM}} (\Psi)$ for $\alpha = 0.1$ (red triangle points) and $\alpha = 0.5$ (blue diamond points) by increasing the bipartition size $N_B$ from $2$ to $4$. Red and blue dashed lines interpolate between the triangle and diamond points, respectively. Solid lines refer to the exact values of $E_\alpha(\Psi)$ for $\alpha = 0.1$ (red solid lines) and $\alpha = 0.5$ (blue solid lines). Black dashed lines refer to the exact values of $E_2(\Psi)$.}
\label{Fig:Suppl2}
\end{figure}

As physical examples, we consider two different models in a $1$-D spin system. We focus on the case of measuring the REE of order $2$ after applying POVMs, which leads to the following REE estimation bound:
$$
E_\alpha(\Psi) \geq \hat{E}_\alpha (\Psi) := \left( \frac{\alpha}{1-\alpha} \right) \log \langle e^{  \left( \frac{1-\alpha}{\alpha} \right)  E_2 } \rangle,
$$
where $\langle e^{  \left( \frac{1-\alpha}{\alpha} \right)  E_2 } \rangle = \sum_m p_m e^{  \left( \frac{1-\alpha}{\alpha} \right)  E_2 (\Psi_m)}$ can be obtained from outcome statistics $\{ p_m, E_2(\Psi_m) \}$.
\subsection{A. Heisenberg model}
First, we consider a Heisenberg model whose Hamiltonian is given by 
$$
H = - J \sum_{j=1}^{N} {\vec \sigma}^{(j)} \cdot \vec \sigma^{(j+1)}
$$
with periodic boundary condition $\vec\sigma^{(N+1)} = \vec\sigma^{(1)}$. Let us suppose that the system is initially prepared in the Neel state $\left | \downarrow \uparrow \downarrow \cdots \uparrow \right\rangle$, which does not have entanglement. As the system undergoes the time evolution, the state $\ket\Psi =  e^{- i H \tau /\hbar} \left | \downarrow \uparrow \downarrow \cdots \uparrow \right\rangle$ becomes entangled after some $\tau$.

We first investigate entanglement between subsystems of an $N = 8$ spin system, divided into $N_A = 6$ and $N_B=2$ after time evolution $0 \leq J \tau \leq 10$ in units of $\hbar = 1$. We employ different types of POVMs acting on single and two-spin sites in the subsystem. For single-spin measurements, we apply two-different types of measurements on the $k$th spin: 1) Projection measurements in the Pauli-$X$ ($\{ \left| \uparrow_x \right\rangle_k \left\langle \uparrow_x \right|, \left| \downarrow_x \right\rangle_k \left\langle \downarrow_x \right| \}$), Pauli-$Y$ ($\{ \left| \uparrow_y \right\rangle_k \left\langle \uparrow_y \right|, \left| \downarrow_y \right\rangle_k \left\langle \downarrow_y \right| \}$), or Pauli-$Z$  ($\{ \left| \uparrow_z \right\rangle_k \left\langle \uparrow_z \right|, \left| \downarrow_z \right\rangle_k \left\langle \downarrow_z \right| \}$) basis
and 2) POVMs along the z-axis $\{ \Pi^{(k)}_{\uparrow_z}(\epsilon), \Pi^{(k)}_{\downarrow_z}((\epsilon) \}$, where $\Pi^{(k)}_{\uparrow_z}(\epsilon) = (1-\epsilon) \left| \uparrow_z \right\rangle_k \left\langle \uparrow_z \right| + \epsilon \left| \downarrow_z \right\rangle_k \left\langle \downarrow_z \right| $ and $\Pi^{(k)}_{\downarrow_z}(\epsilon)= \epsilon \left| \uparrow_z \right\rangle _k \left\langle \uparrow_z \right| + (1- \epsilon) \left| \downarrow_z \right\rangle_k \left\langle \downarrow_z \right| $. For both protocols, we evaluate the maximum value of $\hat{E}_\alpha(\Psi)$ among the measurements on $k = 1, \cdots, N$. For two-spin measurements, we optimise $\hat{E}_\alpha(\Psi)$ over all possible dichotomic POVMs. Note that in this example, two-spin POVMs are sufficient to express all the LOCC operation acting on the subsystem $B$ as $N_B=2$.

Figure \ref{Fig:Suppl1} shows that $\hat{E}^{\rm two-spin}_{\alpha, {\rm POVM}}(\Psi)$ obtained from the optimal two-spin dichotomic POVMs always provides an improved lower bound compared to the direct measurement of $E_2(\Psi)$ for all values of $0 \leq \alpha \leq 2$ and $0 \leq J \tau \leq 10$. In our example, $\hat{E}^{\rm two-spin}_{\alpha, {\rm POVM}}(\Psi)$ is close to the exact value $E_\alpha(\Psi)$ for $\alpha = 0.01$ and $\alpha = 0.1$, while some amount ($|\hat{E}^{\rm two-spin}_{\alpha, {\rm POVM}}(\Psi) - E_2(\Psi) | \sim 0.1$) of improvement can be seen for $\alpha = 0.5$ compared to the direct measurement of $E_2(\Psi)$. The gap between $\hat{E}^{\rm two-spin}_{\alpha, {\rm POVM}}(\Psi) $ and $E_2(\Psi)$ becomes smaller when $\alpha$ increases as $\hat{E}_\alpha(\Psi)$ is a monotone decreasing function on $\alpha$, and we note that $\hat{E}^{\rm two-spin}_{\alpha, {\rm POVM}}(\Psi)$ reduces to $E_2(\Psi)$ when $\alpha =2$. 

On the other hand, when we perform POVMs on a single spin, the improvement of the bound is observed only for the small values of $\alpha$. In our example, the gaps  between $\hat{E}_\alpha(\Psi)$, $E_\alpha(\Psi)$ and $E_2(\Psi)$ strongly depend on $J \tau$ for both single-spin Pauli and POVM measurements, where only a limited set of SLOCC can be implemented by those measurements. This implies that POVMs involving more than a single-spin are demanded for a better REE estimation performance in general cases. The costs for achieving improved REE bounds with general POVMs are the computational costs of optimising POVMs in a higher-dimensional Hilbert space and the implementation of collective measurements on multiple-spins, which could be more challenging compared to the single-spin measurements.

We study how the REE estimation bound changes by varying the bipartition size $N_B$. In the case of optimising over the dichotomic POVMs on the subsystem $B$, $\hat{E}_{\alpha, {\rm POVM}} (\Psi)$ shows better performances as the number of spins in the bipartition increases. Figure \ref{Fig:Suppl2} shows that the gap between $\hat{E}_{\alpha, {\rm POVM}} (\Psi)$ and $E_2(\Psi)$ becomes larger with increasing partition size $N_B$. Furthermore, when $N_B$ becomes larger, the behaviours of $\hat{E}_{\alpha, {\rm POVM}} (\Psi)$ with respect to $J \tau$ tend to follow the exact values $E_\alpha(\Psi)$ rather than $E_2(\Psi)$.

\begin{figure}[b]
\includegraphics[width=.32\linewidth]{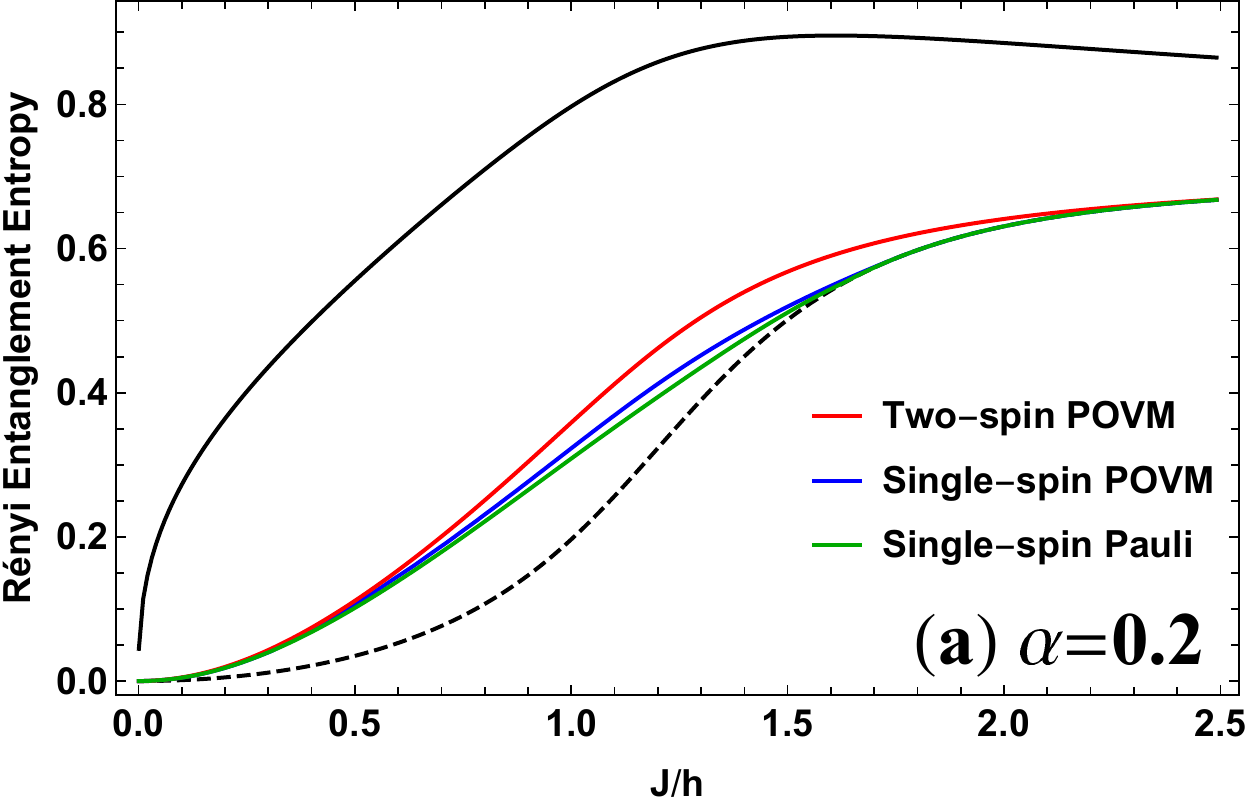}
\includegraphics[width=.32\linewidth]{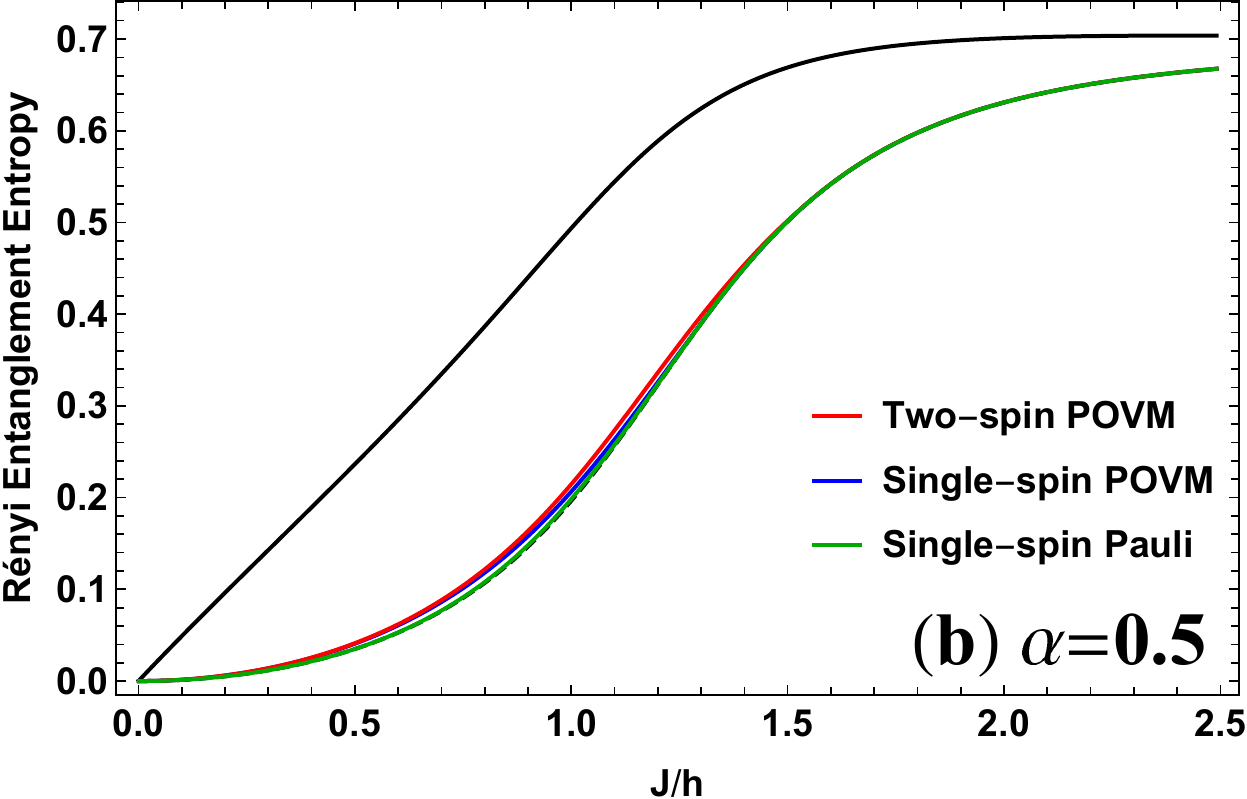}
\includegraphics[width=.32\linewidth]{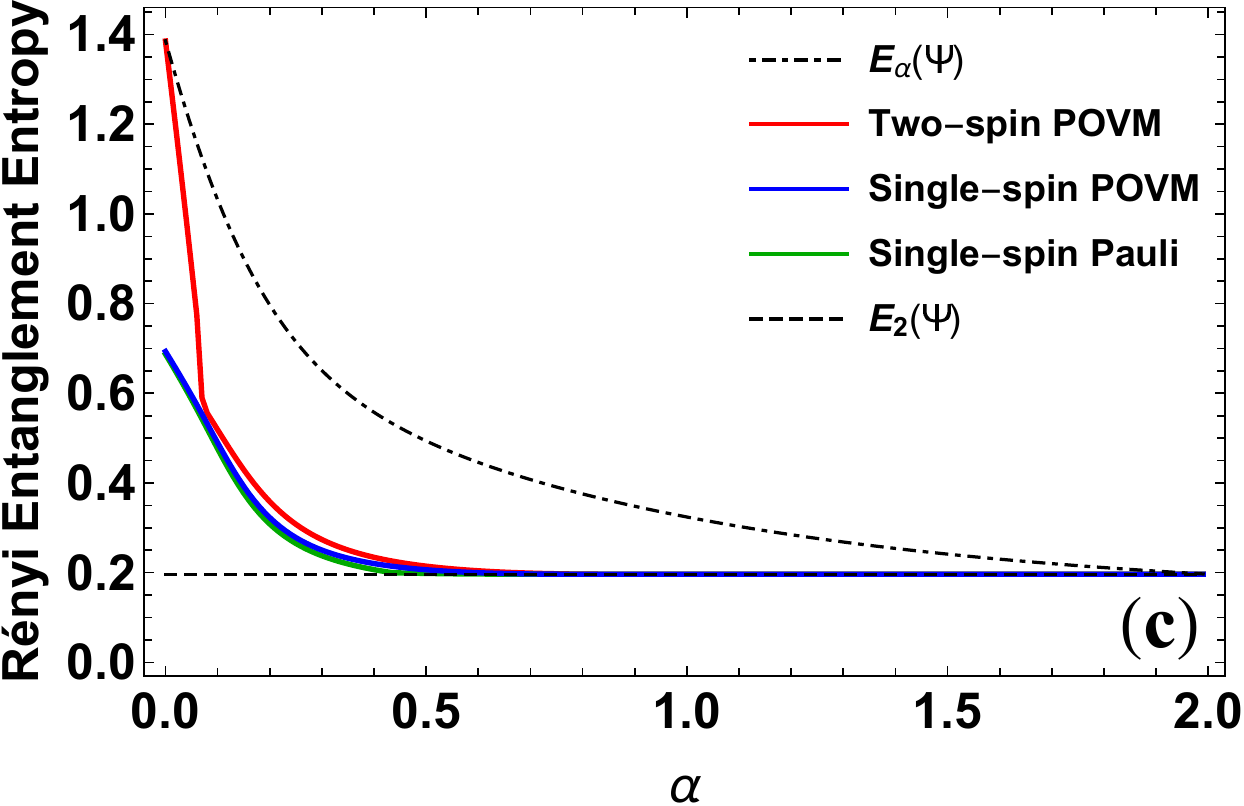}
\caption{$E_\alpha(\Psi)$ of the ground state of $H_{\rm Ising}$ at the critical point $J/h = 1$ and its lower bounds $\hat{E}_\alpha(\Psi)$ for (a) $\alpha = 0.2$ and (b) $\alpha = 0.5$. (c)  $E_\alpha(\Psi)$ and $\hat{E}_\alpha(\Psi)$ for $0 \leq \alpha \leq 2$. The subsystems are chosen as $N_A=6$ and $N_B=2$, where $N = N_A + N_B = 8$. Three different types of POVMs, single-spin Pauli, and single-spin POVM, and two-spin POVM are applied to obtain the lower bounds. $E_\alpha(\Psi)$ is obtained from direct diagonalisation of the Ising Hamiltonian $H_{\rm Ising}$ (black-solid lines). The black-dotted lines refer to $E_2(\Psi)$.}
\label{Fig:Suppl3}
\end{figure}

\subsection{B. Ising model}
We consider another physical model, the transverse Ising model with the following Hamiltonian
$$
H_{\rm Ising} = -h \sum_{j=1}^{N} \sigma_z^{(j)} - J \sum_{j=1}^{N-1} {\sigma}^{(j)}_x \sigma^{(j+1)}_x,
$$
where $J$ is the interaction strength and $h$ is the external magnetic field strength. Phase transition of this system occurs at $J/h = 1$. The ground state has no entanglement when $J=0$, while the system becomes more entangled as $J$ increases until it reaches the critical point $J/h=1$. After passing the critical point, entanglement tends to saturate at $\log 2$ as $J \rightarrow \infty$.

We first investigate entanglement for the ground state of $H_{\rm Ising}$ with $N_A=6$ and $N_B = N - N_A = 2$ by increasing the interaction strength $J/h$. Figure \ref{Fig:Suppl3} shows that $E_\alpha(\Psi)$ for $\alpha = 0.2$ increases faster than $E_2(\Psi)$ before the system reaches the critical point and changes more gradually than the higher order REE near the critical point. We also note that the single-spin POVMs are enough to observe this, while optimisation over POVMs only gives small improvement of the bound. However, $\alpha$ becomes larger as $\alpha = 0.5$, only a small amount of improvement in the REE estimation can be obtained. Furthermore, due to the finite size effect, it is hard to recognise the phase transition at $J/h =1$ based on the REEs. 

\begin{figure}[t]
\includegraphics[width=.43\linewidth]{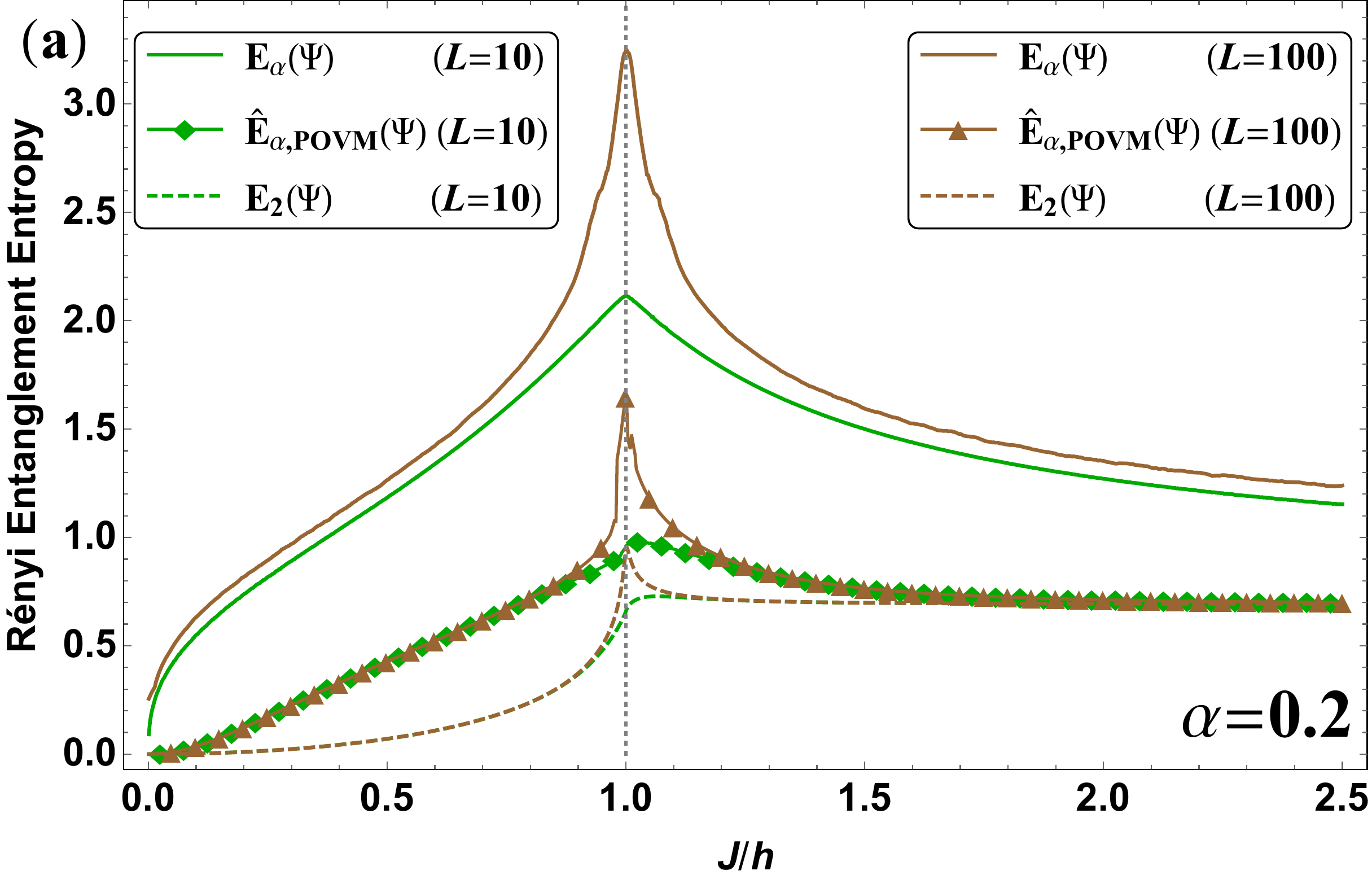}\qquad
\includegraphics[width=.43\linewidth]{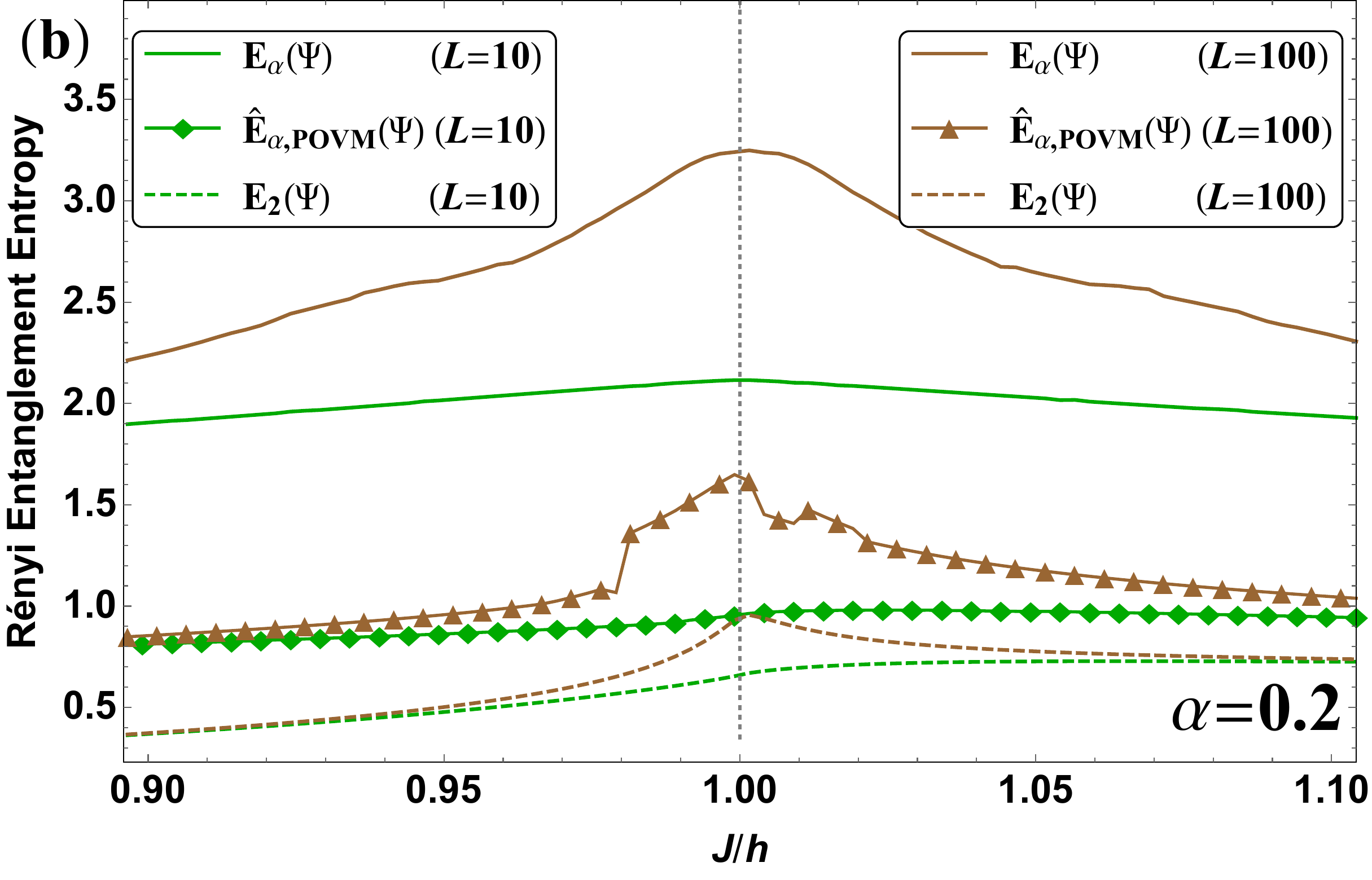}\\~\\
\includegraphics[width=.43\linewidth]{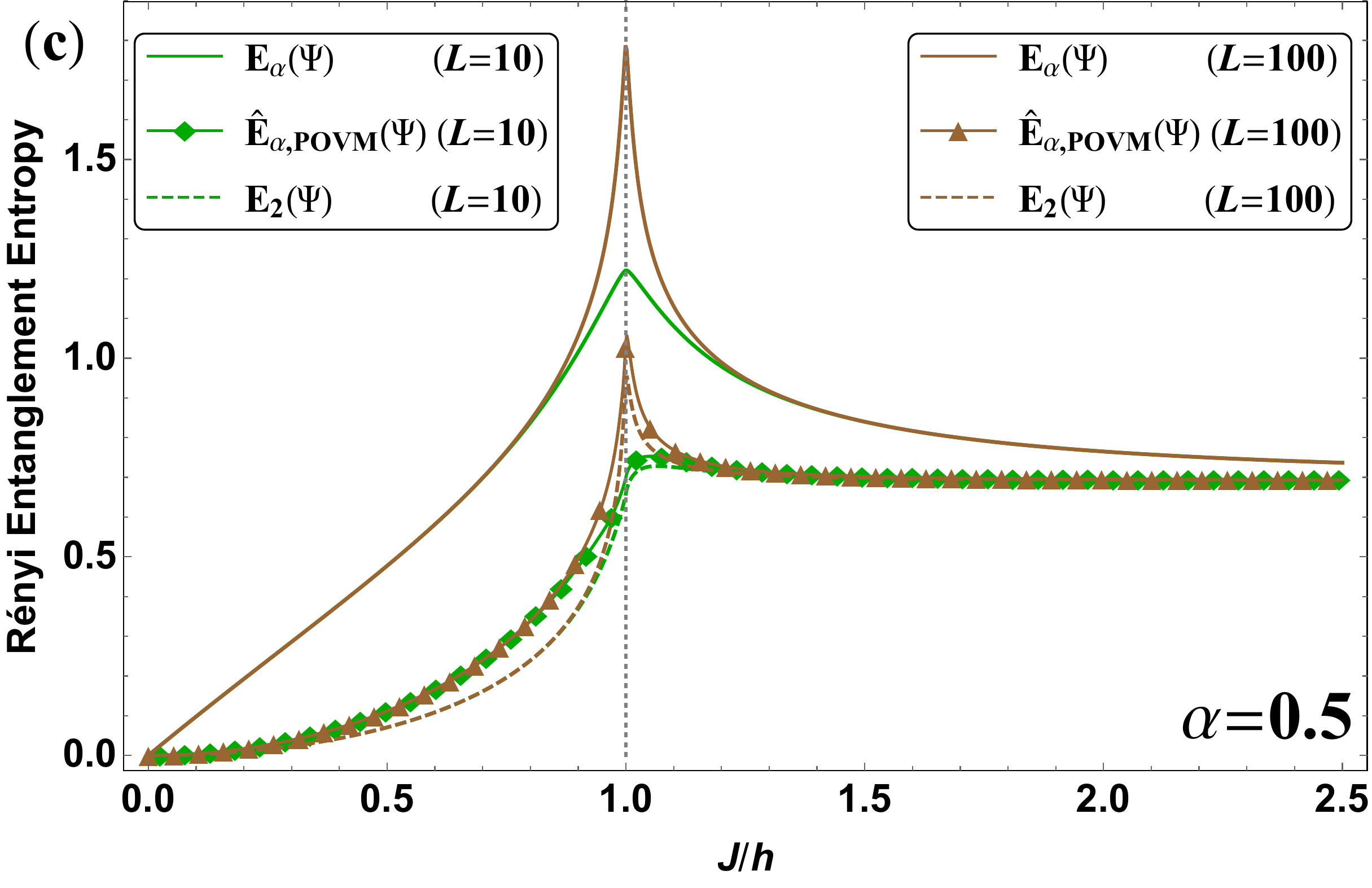}\qquad
\includegraphics[width=.43\linewidth]{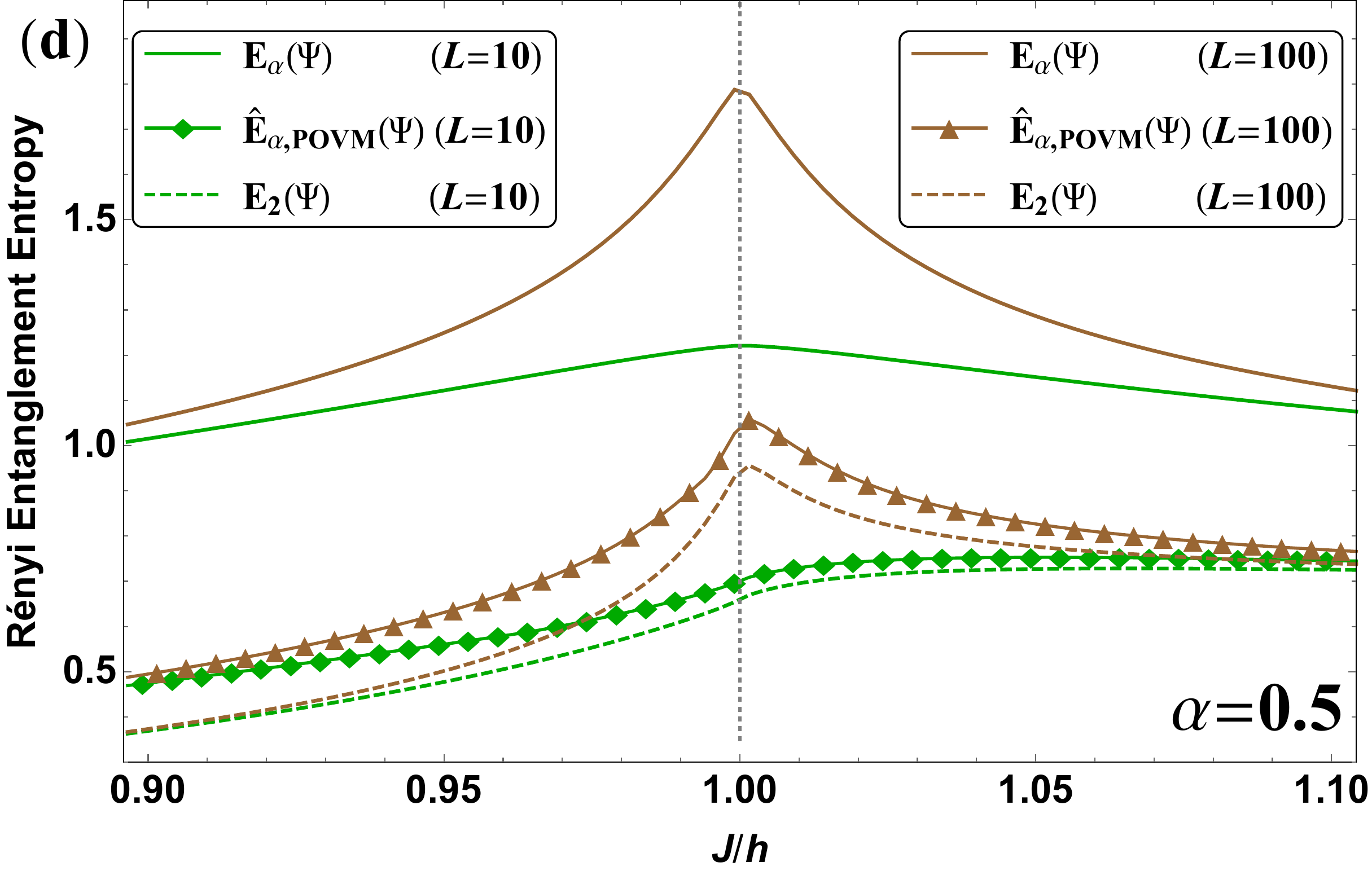}\\~\\
\includegraphics[width=.43\linewidth]{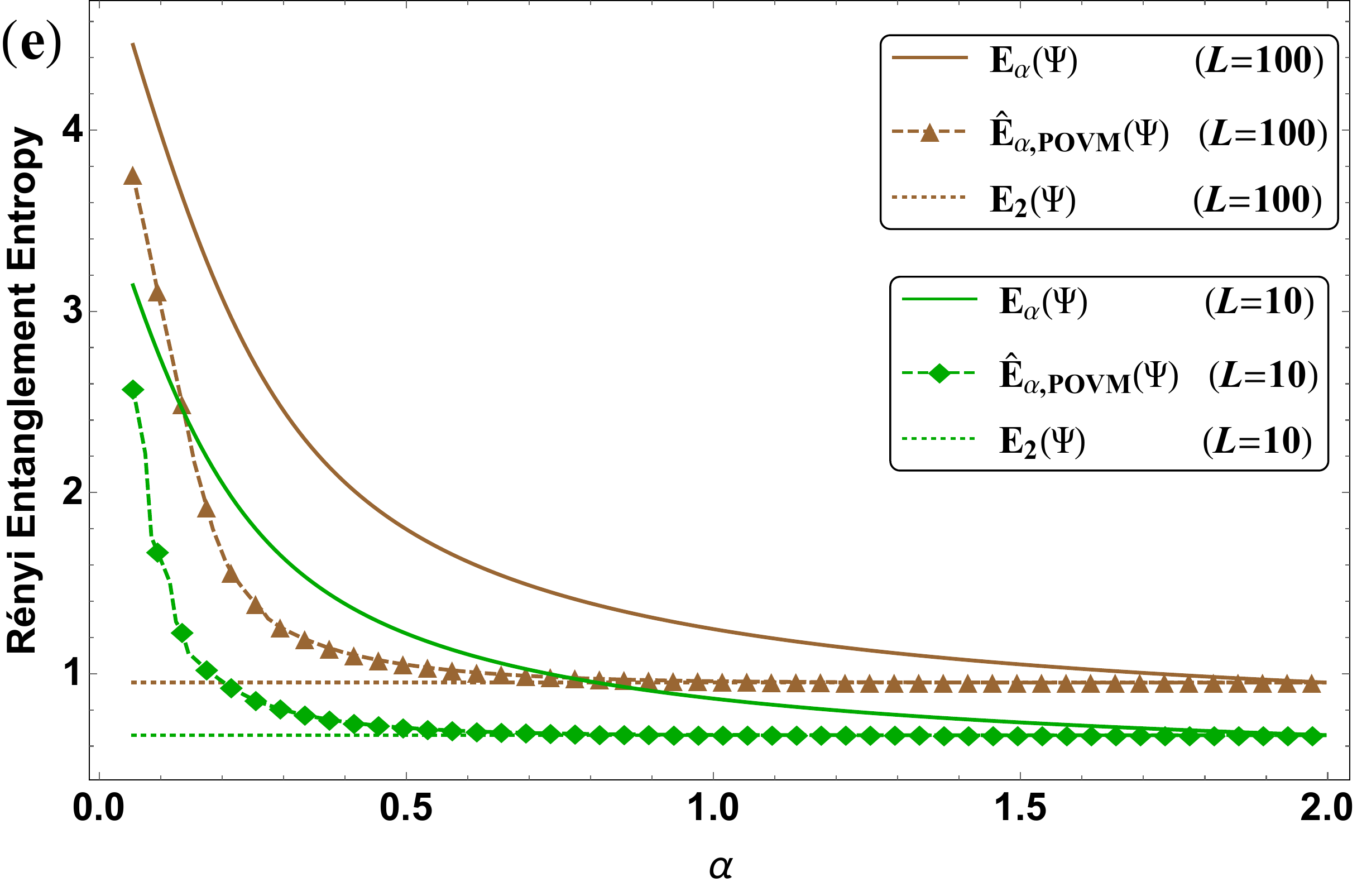}\qquad
\includegraphics[width=.43\linewidth]{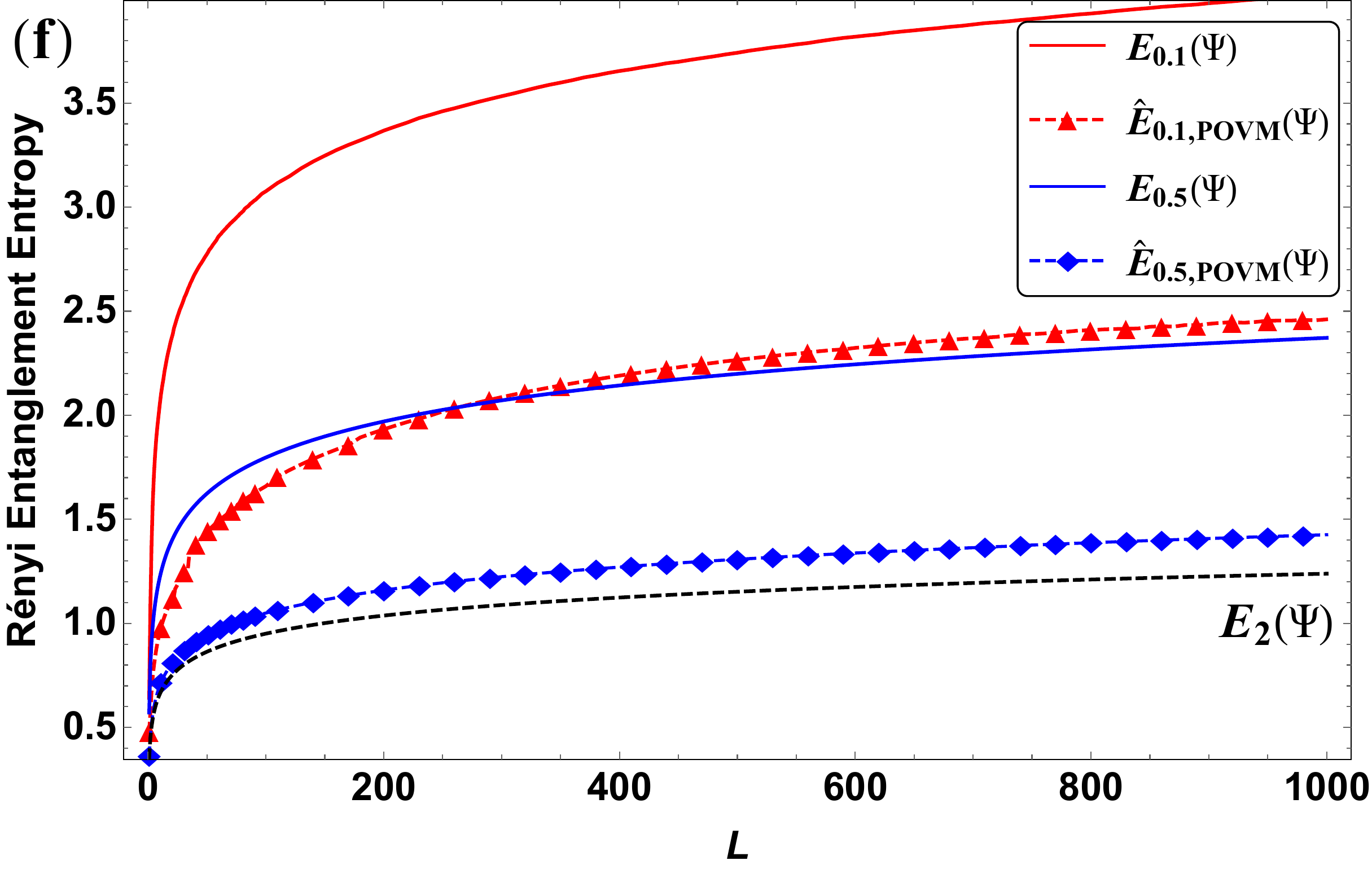}
\caption{REEs $E_\alpha(\Psi)$ for the ground state of the Ising Hamiltonian $H_{\rm Ising}$ in the thermodynamic limit ($N \rightarrow \infty$) and their lower bounds obtained from $\hat{E}_{\alpha, {\rm POVM}}(\Psi)$ by optimising over dichotomic POVMs on the local system $\rho_L = \bigotimes_{l=0}^{L-1} \rho_l$. The REEs and their bounds are obtained by varying $J/h$ from $0$ to $2.5$ for (a) $\alpha = 0.2$, and (c) $\alpha=0.5$, while (b) and (d) present those for the narrow range of $0.9 \leq J/h \leq 1.1$. (e) The REEs and their bounds for various values of $0 \leq \alpha \leq 2$ at the critical point $J/h = 1$. (f) Behaviours of the REE bounds by increasing the bipartition size $L$. The fermionic modes only with $\nu_l \leq 1- 10^{-10}$ have been taken into account for evaluating $\hat{E}_{\alpha, {\rm POVM}} (\Psi)$.}
\label{Fig:Suppl4}
\end{figure}

In order to handle these problems, we study the ground state REEs of the Ising model in the thermodynamic limit $N\rightarrow \infty$. The entanglement spectrum for the ground state of an XY Hamiltonian,
$$
H_{XY} = -\sum_{j=0}^{N-1} \left( \frac{a}{2}\left[ (1+\gamma) {\sigma}^{(j)}_x \sigma^{(j+1)}_x+ (1-\gamma) {\sigma}^{(j)}_y \sigma^{(j+1)}_y\right] +\sigma_x^{(j)} \right),
$$
can be evaluated in the thermodynamic limit \cite{Vidal03, Latorre04}. The ground state of the system can be characterised by using Majorana operators,
$$
c_{2l} = \left( \prod_{j=0}^{l-1}\sigma_z^{(l)}  \right) \sigma_x^{(l)}~{\rm and}~c_{2l+1} = \left( \prod_{j=0}^{l-1}\sigma_z^{(l)}  \right) \sigma_y^{(l)},
$$
satisfying the anticommutation relations $\{ c_k, c_l \} = 2 \delta_{kl}$. The reduced density matrix of the ground state having partition size $L$ can then be characterised by $\langle c_k c_l \rangle = \delta_{kl} + i \Gamma_{kl}$ as higher order moments can be evaluated by using Wick's theorem. In the thermodynamic limit $N \rightarrow \infty$, the bipartite partition with length $L$ can be evaluated by the following $2L \times 2L$ matrix:
$$
\Gamma_L = \left(
\begin{matrix}
\Pi_0 & \Pi_1 & \cdots & \Pi_{L-1} \\
\Pi_{-1} & \Pi_0 & \cdots & \Pi_{L-2} \\
\vdots & \vdots & \ddots & \vdots  \\
\Pi_{1-L} & \Pi_{2-L} & \cdots & \Pi_0
\end{matrix}
\right),
\qquad
\Pi_l = \left(
\begin{matrix}
0 &g_l \\
-g_{-l} & 0
\end{matrix}
\right),
$$
with real coefficients
$$
g_l = \frac{1}{2\pi} \int_0^{2\pi} d\phi e^{-i l\phi} \frac{a \cos\phi -1 - i a \gamma \sin \phi}{|a \cos\phi -1 - i a \gamma \sin \phi|} = \frac{1}{2\pi} \int_0^{2\pi} d\phi e^{-i l\phi} \frac{ (J/h) \cos\phi -1 - i (J/h)\sin \phi}{| (J/h) \cos\phi -1 - i (J/h) \sin \phi|}
$$
for the Ising model $\gamma = 1$ and $a = J/h$. Note that $\Gamma_L$ can be converted into a block-diagonal form
$$
\tilde\Gamma_L = V \Gamma_L V^T = \bigoplus_{l=0}^{L-1} \left(
\begin{matrix}
0&\nu_l \\
-\nu_l &0
\end{matrix}
\right),
$$
where $V$ is an element of $SO(2L)$. By using a new set of Majorana operators $d_m = \sum_{n=0}^{2L-1} V_{mn} c_n$, we can define fermionic modes
$$
b_l = (d_{2l} + i d_{2l+1})/\sqrt{2},
$$
satisfying $\langle b_l \rangle = 0 = \langle b_k b_l \rangle$ and $\langle b_k^\dagger b_l \rangle = \delta_{kl} \left( \frac{1+ \nu_l}{2} \right)$. Finally, the density matrix of the partition with length $L$ can be expressed as a product of the mixed state of each mode $l$,
$$
\rho_L = \bigotimes_{l=0}^{L-1} \rho_l,
$$
where $\rho_l$ has eigenvalues $\lambda(\rho_l) = \left\{ \frac{1+\nu_l}{2},  \frac{1-\nu_l}{2} \right\}$.
Thus, the REE of the bipartition can be evaluated as
$$
E_\alpha(\Psi) = S_\alpha(\rho_L) = \sum_{l=1}^{L-1} S_\alpha (\rho_l).
$$

In order to estimate the REEs with $\hat{E}_\alpha(\Psi)$, we apply POVMs acting on $\rho_L = \bigotimes_{l=0}^{L-1} \rho_l$. Although the $L$ fermionic modes are uncorrelated, we note that this does not mean that the local system does not factorise into local Hilbert spaces for $L$-spins as each fermionic mode has a non-local structure. Thus, collective measurements on the spin in the subsystem become a general form of such POVMs. Nevertheless, it is always possible to optimise over POVMs acting on the local system by expressing them in terms of the fermionic operators.

Figure~\ref{Fig:Suppl4} shows how the REEs and their bounds behave by varying the system parameter $J/h$ and the bipartition size $L$. When the bipartition size $L$ becomes larger, transition in the REEs near the critical point becomes more significant. We note that the REEs of lower order increase more sharply as the system approaches the critical point. This can be understood as when the system is near the critical point, some of the density matrices corresponding to the fermionic mode $l$ can have $\nu_l$  close to $1$ \cite{Latorre04}, which can be more sensitively captured by lower order REEs. The lower bound of REE $\hat{E}_{\alpha, {\rm POVM}}$ is obtained by employing optimised dichotomic POVMs on the local system. The bounds $\hat{E}_{\alpha, {\rm POVM}} (\Psi)$ can estimate REEs better than directly measuring $E_2(\Psi)$ without applying the POVMs. Similar to the case of the Heisenberg model, the gap between those bounds $| \hat{E}_{\alpha, {\rm POVM}} (\Psi) - E_2(\Psi)|$ near the critical point becomes larger when the bipartition size $L$ becomes larger, where more spins are involved in the POVMs.

\section{IX. Estimation of the REE under decoherence}
We show that our REE estimation protocol can be applied when the pure bipartite state $\ket{\Psi}_{AB}$ undergoes a decoherence process and becomes a mixed state
$$
\tilde\rho_{AB} = (1- z) \ket{\Psi}_{AB} \bra{\Psi} +  z \sigma_{AB},
$$
where $\sigma_{AB}$ can be determined based on a decoherence model. Here, $z$ indicates the degree of decoherence, as $z=0$ refers to the decoherence-free case $\ket{\Psi}_{AB}$, while the system is fully-decohered to $\sigma_{AB}$ when $z=1$.
After applying POVMs on the subsystem that can be represented using Kraus operators $\{ K_m \}$, the quantum state transforms into
$$
\tilde\rho_{AB} \rightarrow \{ \tilde q_m, \tilde\rho_{AB}^m\},
$$
where the state $\tilde\rho_{AB}^m$ refers to the quantum state that we obtain from the measurement outcome $m$. We then express $\tilde \rho_{AB}^m$ in terms of the decoherence-free term $\ket{\Psi_m}$ and the fully-decohered term $\sigma_{AB}$ as
$$
\begin{aligned}
\tilde\rho_{AB}^m &= \frac{1}{\tilde q_m} \left[ (1-z) K_m \ket{\Psi}_{AB} \bra{\Psi} K_m^\dagger + z K_m \sigma_{AB} K_m^\dagger \right]\\
&= \frac{1}{\tilde q_m} \left[  (1-z) p_m \ket{\Psi_m}_{AB} \bra{\Psi_m} + z r_m \sigma^m_{AB} \right].
\end{aligned}
$$
Here $\ket{\Psi_m} = K_m \ket{\Psi} / \sqrt{p_m} $ with $p_m =  \bra{\Psi} K_m^\dagger K_m \ket{\Psi}_{AB}$ and $\sigma_{AB}^m = K_m \sigma_{AB} K_m^\dagger / r_m$ with $r_m = \Tr \left[ \sigma_{AB} K_m^\dagger K_m \right]$ satisfying $\tilde q_m= (1-z)p_m + z r_m$. Subsequently, the reduced density matrix for the local system $B$ can be described as
$$
\tilde\rho_{B}^m = \frac{1}{\tilde q_m} \left[  (1-z) p_m \rho_B^m+ z r_m \sigma^m_{B} \right],
$$
where we denote $\rho_B^m = \Tr_A \ket{\Psi_m}_{AB} \bra{\Psi_m}$ and $\sigma_B^m = \Tr_A \sigma_{AB}^m$.

\begin{figure}[b]
\includegraphics[width=.45\linewidth]{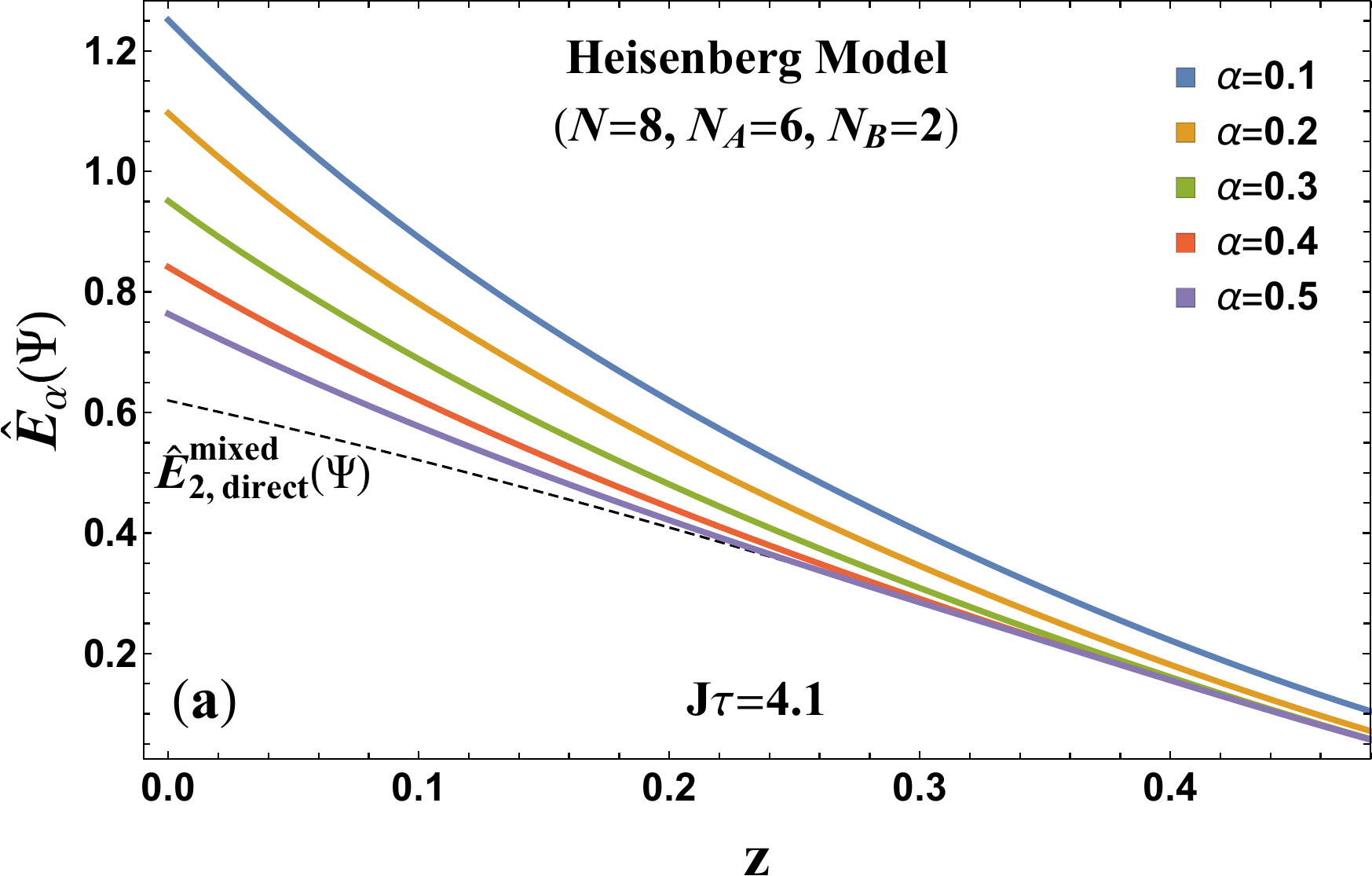} \qquad
\includegraphics[width=.45\linewidth]{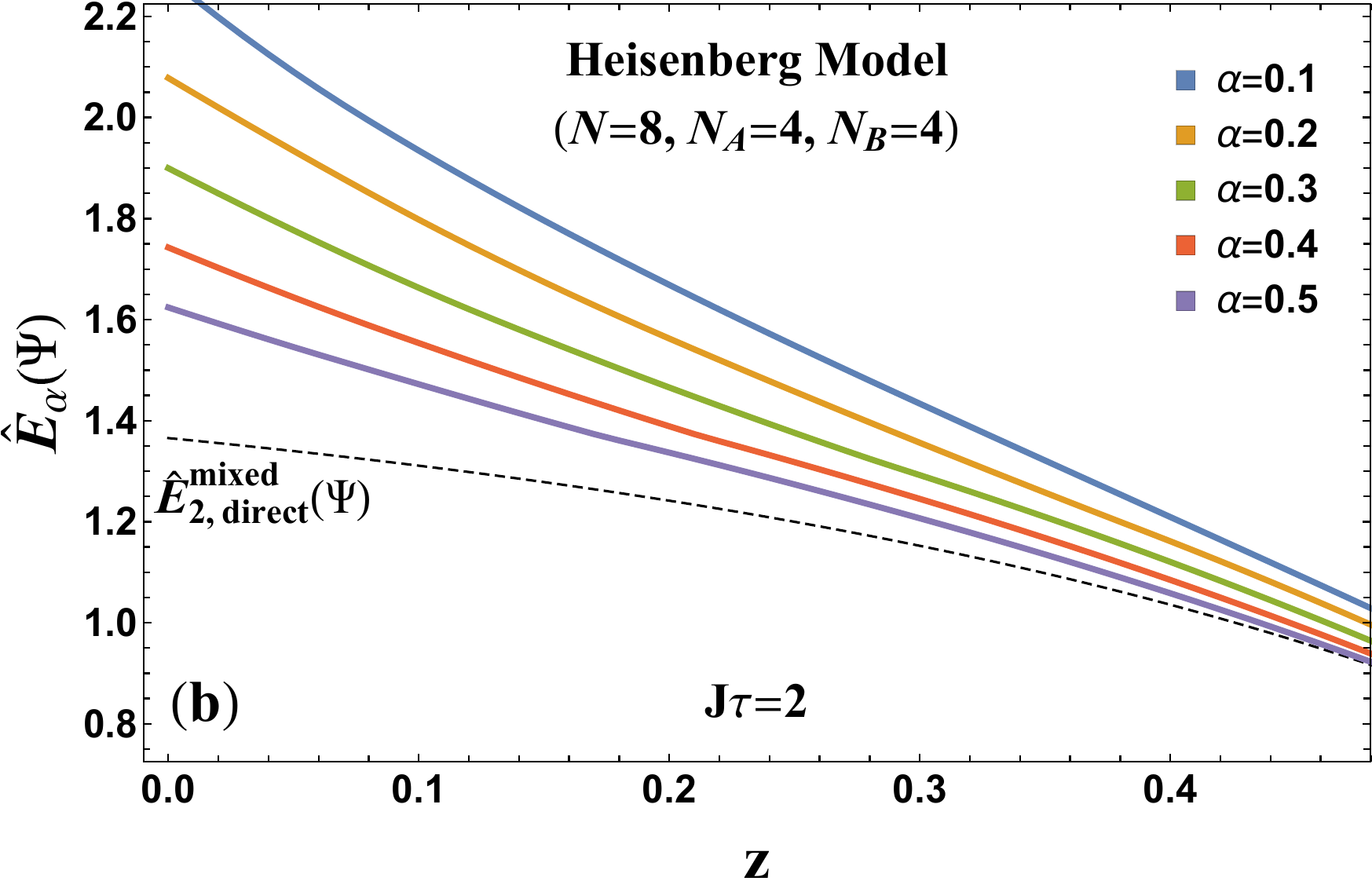}\\~\\~\\
\includegraphics[width=.45\linewidth]{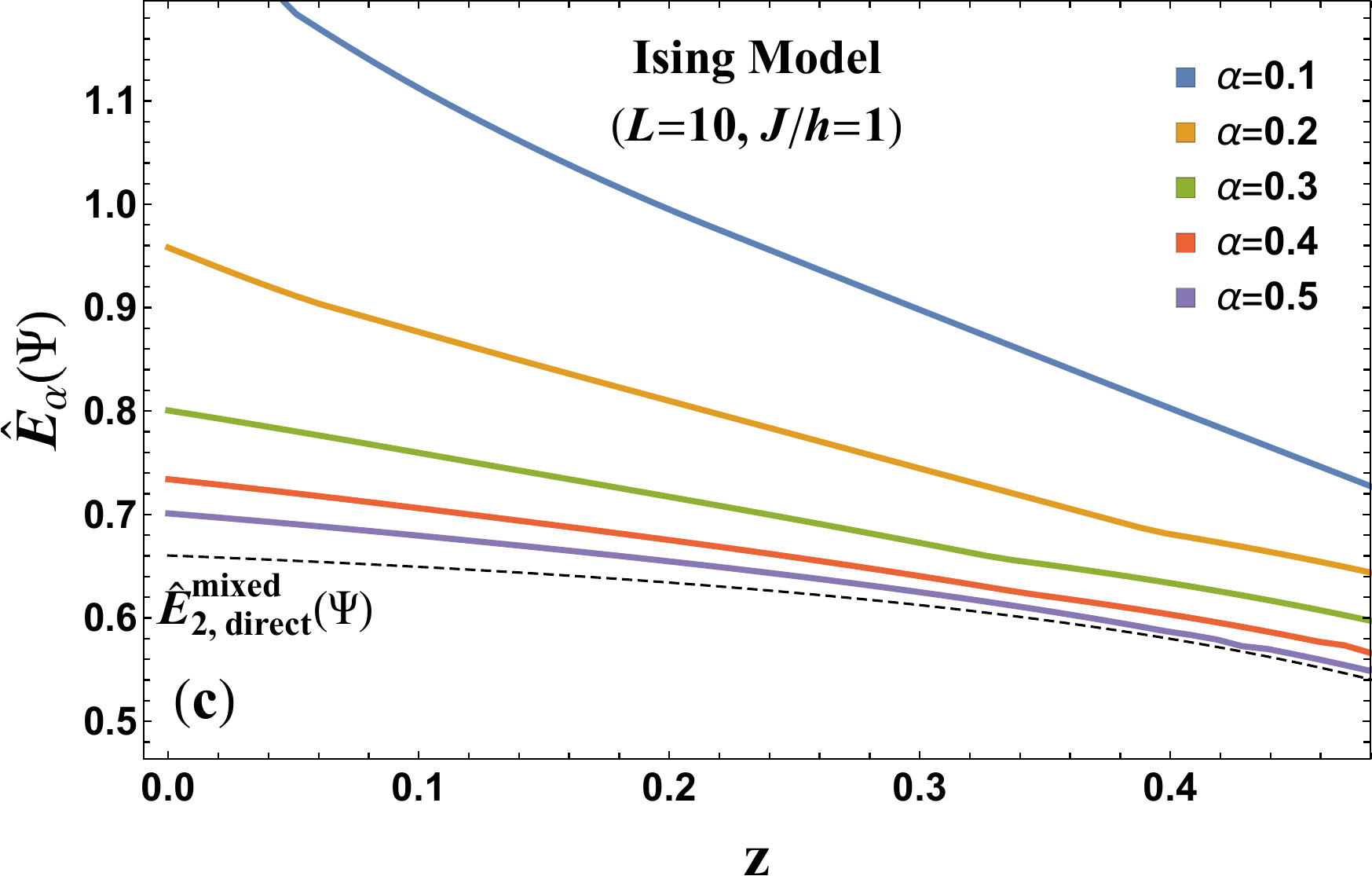}\qquad
\includegraphics[width=.45\linewidth]{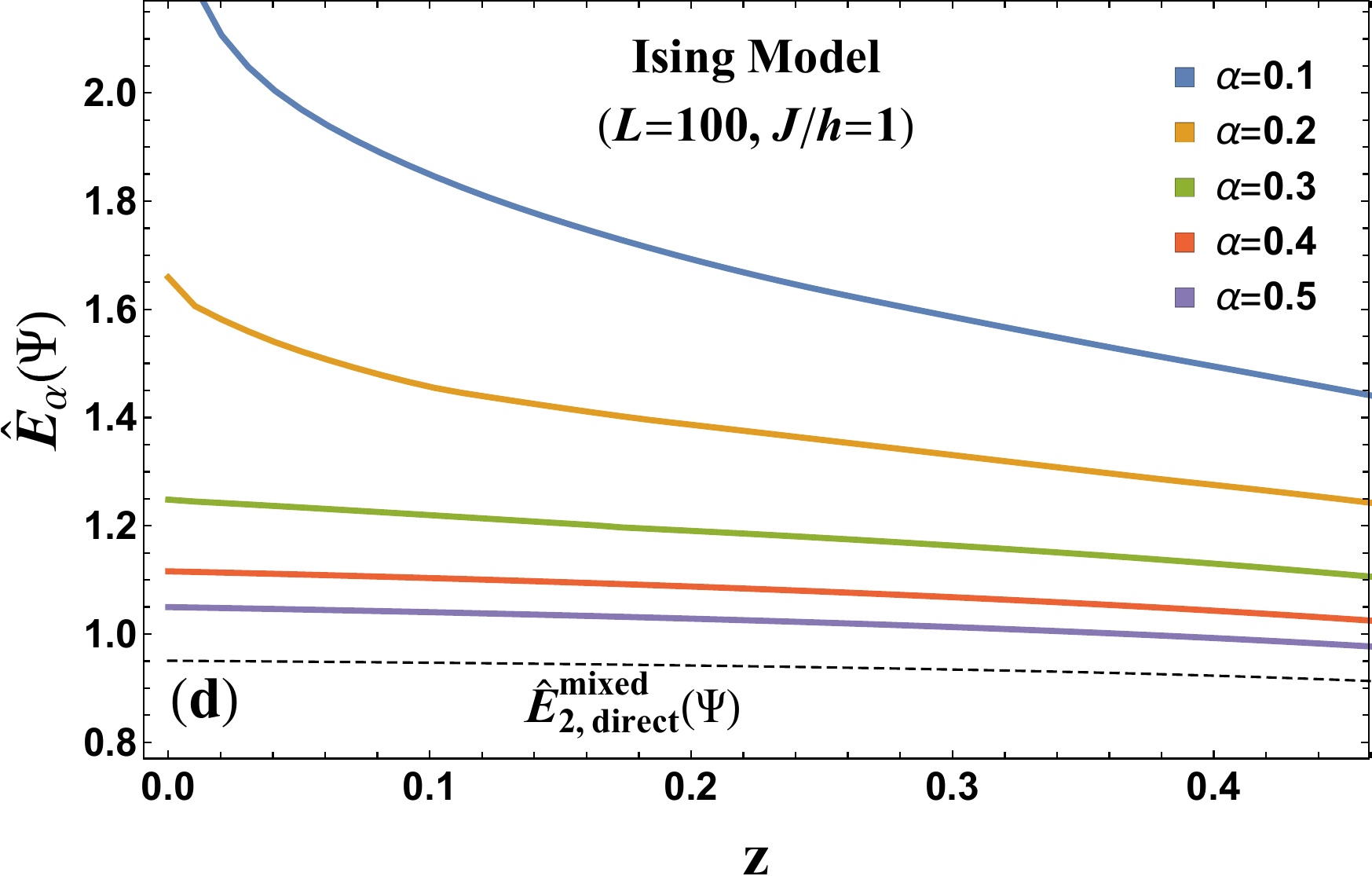}
\caption{$\hat{E}^{\rm mixed}_{\alpha, {\rm POVM}} (\Psi)$ under the global depolarising channel for the Heisenberg Model and Ising model with various parameters. The Heisenberg model refers to the state $\ket\Psi =  e^{- i H \tau /\hbar} \left | \downarrow \uparrow \downarrow \cdots \uparrow \right\rangle$ discussed in the previous section, while the Ising model refers to the ground state in the thermodynamic limits, $N\rightarrow \infty$ with bipartition size $L$. The local measurements applied to the subsystem are optimised over dichotomic POVMs for both cases.}
\label{Fig:Suppl5}
\end{figure}

We recall that the pure state REE estimation is given by
$$
E_\alpha(\Psi) \geq \hat{E}_\alpha(\Psi) = \left( \frac{\alpha}{1-\alpha} \right) \log \left[ \sum_m p_m e^{ \left( \frac{1-\alpha}{\alpha} \right) E_2(\Psi_m)} \right] = \left( \frac{\alpha}{1-\alpha} \right) \log \left[ \sum_m p_m e^{ \left( \frac{1-\alpha}{\alpha} \right) S_2(\rho^m_B)} \right],
$$
where $p_m = \frac{ \tilde q_m - z r_m}{1-z}$ can be obtained in experiments by determining $\tilde q_m$ and $z$ and evaluating $r_m$ from the known expression of $\sigma_{AB}$. We note that $S_2(\rho_B^m)$ can be estimated from the experimentally accessible value $S_2(\tilde\rho_B^m)$ by utilising the following Lemma recently shown by Hanson and Datta \cite{Hanson17}.
\begin{lemma} [Uniform continuity bound for the R\'enyi entropy \cite{Hanson17}] Suppose that the trace distance between two quantum states $\rho$ and $\sigma$ in a finite-dimensional Hilbert space $\cal H$ with $d = \dim{\cal H}$ is bounded as $\frac{1}{2} \|\rho - \sigma\|_1 \leq \epsilon$. The difference between the R\'enyi entropies of two states is then upper bounded by
\begin{equation}
| S_\alpha(\rho) - S_\alpha(\sigma)| \leq f_\alpha(\epsilon, d) =
\begin{cases}
\frac{1}{1-\alpha} \log \left[ (1-\epsilon)^\alpha +(d-1)^{1-\alpha} \epsilon^\alpha \right] &\epsilon \leq 1- \frac{1}{d} \\
\log d& \epsilon > 1- \frac{1}{d}
\end{cases}.
\end{equation}
\label{Lemma:Datta}
\end{lemma}
From Lemma~\ref{Lemma:Datta}, we obtain
$$
S_2(\rho^m_B) \geq S_2(\tilde\rho^m_B) - f_2 (\epsilon_m, d),
$$
where $\frac{1}{2} \|\rho_B^m - \tilde\rho_B^m \|_1 \leq \epsilon_m$. We can take $\epsilon_m$ to be $z (r_m / \tilde q_m)$ by the following inequality
$$
\begin{aligned}
\frac{1}{2} \|\rho_B^m - \tilde\rho_B^m \|_1 &= \frac{1}{2 \tilde q_m} \| \tilde q_m \rho_B^m - \left[ (1-z) p_m \rho_B^m + z r_m \sigma_B^m \right] \|_1\\
&= \left( \frac{z r_m}{\tilde q_m} \right) \frac{1}{2} \| \rho_B^m - \sigma_B^m \|_1 \\
&\leq \frac{z r_m}{\tilde q_m}.
\end{aligned}
$$
By assuming that $\epsilon_m = z (r_m / \tilde q_m) \leq 1- \frac{1}{d}$, we then have
$$
f_2(\epsilon_m, d) = - \log \left[ (1-\epsilon_m)^2 + \frac{\epsilon_m^2}{d-1}\right] \leq - 2 \log (1- z (r_m / \tilde q_m)).
$$
By combining these results, the REE of the uncontaminated pure state $\ket{\Psi}_{AB}$ is lower bounded by
\begin{equation}
\begin{aligned}
E_\alpha(\Psi) &\geq \left( \frac{\alpha}{1-\alpha} \right) \log \left[ \sum_m p_m e^{ \left( \frac{1-\alpha}{\alpha} \right) \left[ S_2(\tilde\rho^m_B) - f_2(\epsilon_m,d) \right]} \right] =: \hat{E}^{\rm mixed}_\alpha (\Psi).
\end{aligned}
\end{equation}
For the special case of $\alpha = 2$, the REE of $\Psi$ can be directly estimated by measuring the R\'enyi entropy $S_2(\tilde\rho_B)$ without applying POVMs as
\begin{equation}
E_2(\Psi) \geq S_2(\tilde\rho_B) - f_2(\epsilon,d) =: \hat{E}_{2, {\rm direct}}^{\rm mixed} (\Psi),
\end{equation}
where $\epsilon = \frac{1}{2} \| \rho_B - \tilde \rho_B \|_1$ with $\rho_B = \Tr_A \ket{\Psi}_{AB} \bra{\Psi}$ and $\tilde\rho_B = \Tr_A \tilde\rho_{AB}$.

We apply the extended protocol to a global depolarising channel, for which case $\sigma_{AB} = {\mathbb 1}_{AB} / d_{AB}$. In this case, $\hat{E}^{\rm mixed}_{\alpha, {\rm POVM}} (\Psi)$ after optimising over dichotomic POVMs on the subsystem can give a better estimation of the REEs compared to $\hat{E}_{2, {\rm direct}}^{\rm mixed} (\Psi)$ obtained by directly measuring the R\'enyi entropy of order $2$. The improvement of the REE bound can be obtained even with the presence of decoherence ($z \lesssim 0.5$) for both the Heisenberg and Ising models studied in the previous sections, regardless of the system parameters and bipartition size (see Fig.~\ref{Fig:Suppl5}).

\section{X. Proof of Proposition 4}
\begin{proposition} Suppose that an initial bipartite mixed state $\rho$ transforms by SLOCC. Then, the following inequality holds
$$
{\rm co}{E}_{(\alpha,s)} (\rho) \geq \frac{1}{s(1-\alpha)}\left[ \langle  e^{ s(1-\alpha) {\rm co}{E}_\gamma} \rangle -1\right]
$$
for $0 < \alpha <1$, $\alpha \leq \gamma$, and $s \leq 1/\alpha$. 
From this, the success probability of raising ${\rm co}E_\alpha$ is upper bounded as 
$$
\begin{aligned}
P\left(  {\rm co}E_\alpha  \geq E_{\rm target} \right)
&\leq \min_{0 \leq  \beta \leq \alpha} \left[ \frac{ \left(\frac{1-\beta}{\beta} \right) {\rm co}{E}_{(\beta,1/\beta)} (\rho) }{e^{\left(\frac{1-\beta}{\beta} \right) E_{\rm target}}-1} \right].
\end{aligned}
$$
\label{Prop:Mixed}
\end{proposition}

\begin{proof}
We first show the inequality condition
\begin{equation}
\label{Eq:Appd2}
{\rm co}{E}_{(\alpha,s)} (\rho) \geq \frac{1}{s(1-\alpha)}\left[ \langle  e^{ s(1-\alpha) {\rm co}{E}_\alpha} \rangle -1\right],
\end{equation}
for $0 < \alpha <1$, $\alpha \leq \beta$, and $s \leq 1/\alpha$.
Let us suppose that $\{ q_\mu^*, \ket{\Psi_\mu^*} \}$, a pure state decomposition of $\rho$, minimising the left-hand-side of the inequality, i.e., ${\rm co}E_{(\alpha,s)} (\rho) = \sum_\mu q_\mu^* E_{(\alpha,s)} (\Psi_\mu^*) = \sum_\mu q_\mu^* \left[ \frac{1}{s(1- \alpha)} ( e^{s(1-\alpha) E_\alpha(\Psi_\mu^*)} -1)\right]$. 

Meanwhile, a general SLOCC protocol described by a coarse-grained LOCC instrument \cite{Chitambar14} can have a coarse-grained outcome $M$, which consists of fine-grained outcomes $m \in M$. In this case, the outcome state can be expressed as ${\cal E}^{(M)}_{\rm SLOCC} (\rho) = \sum_{m \in M} {\cal E}^{(m)}_{\rm SLOCC} (\rho) = p_M \rho_M$, where ${\cal E}^{(m)}_{\rm SLOCC}$ is a pure state SLOCC transformation that maps any pure state into another pure state and $p_M = \Tr[ {\cal E}^{(M)}_{\rm SLOCC} (\rho)] =  \sum_{m \in M} {\rm Tr}[{\cal E}^{(m)}_{\rm SLOCC} (\rho)]$. Then, it is possible to express the outcome state as $\rho_M = \frac{1}{p_M} \sum_{m \in M} {\cal E}^{(m)}_{\rm SLOCC} (\rho)  = \sum_{m \in M} \sum_\mu \left( \frac{q_\mu^* r_{\mu  m} }{p_M} \right) \ket{\Psi_\mu^m}\bra{\Psi_\mu^m}$ with ${\cal E}^{(m)}_{\rm SLOCC} (\ket{\Psi_\mu^*} \bra{\Psi_\mu^*} ) = r_{\mu m} \ket{\Psi_\mu^m} \bra{\Psi_\mu^m} $ satisfying $\sum_m r_{\mu m} = 1$. From the definition of the convex roof construction, we note that  ${\rm co}E_\alpha(\rho_M) \leq \sum_{m \in M} \sum_\mu \left( \frac{q_\mu^* r_{\mu m} }{p_M} \right) E_\alpha (\Psi_\mu^m)$. We then complete the proof as follows:
$$
\begin{aligned}
\frac{1}{s(1-\alpha)}\left[ \langle  e^{ s(1-\alpha) {\rm co}{E}_\alpha} \rangle -1\right] &=\frac{1}{s(1-\alpha)} \left[\sum_M p_M e^{ s (1-\alpha) {\rm co}{E}_\alpha (\rho_M)} -1 \right]\\
&\leq \frac{1}{s(1-\alpha)} \left[ \sum_{\mu,m}  q_\mu^* r_{\mu m} e^{ s (1-\alpha)  E_\alpha (\Psi_\mu^m) } -1 \right]\\
&\leq \frac{1}{s(1-\alpha)} \left[  \sum_\mu  q_\mu^* e^{ s (1-\alpha) E_\alpha (\Psi_\mu^*) } -1 \right]\\
& = {\rm co}E_{(\alpha,s)} (\rho).
\end{aligned}
$$
The first inequality is obtained from the convexity of the exponential function and the second inequality comes from the monotonicity of ${\rm co}E_{(\alpha,s)} (\rho)$. By noting that ${\rm co}{E}_\alpha \geq {\rm co}{E}_\gamma$ for $0 \leq \alpha \leq \gamma$, we obtain
$$
{\rm co}E_{(\alpha,s)} (\rho) \geq \frac{1}{s(1-\alpha)}\left[ \langle  e^{ s(1-\alpha) {\rm co}{E}_\alpha} \rangle -1\right] \geq \frac{1}{s(1-\alpha)}\left[ \langle  e^{ s(1-\alpha) {\rm co}{E}_\gamma} \rangle -1\right],
$$
which completes the proof. By taking $s=1/\alpha$, the probability bound,
$$
P\left(  {\rm co}E_\alpha  \geq E_{\rm target} \right) \leq \min_{0 \leq  \beta \leq \alpha^*} \left[ \frac{ \left(\frac{1-\beta}{\beta} \right) {\rm co}{E}_{(\beta,1/\beta)} (\rho) }{e^{\left(\frac{1-\beta}{\beta} \right) E_{\rm target}}-1} \right],
$$
is obtained by using the same argument in Proposition 2 and that ${\rm co}E_\beta \geq {\rm co}E_\alpha$ for $0 \leq \beta \leq \alpha$.
\end{proof}


\begin{thebibliography}{99}
\bibitem{note1} In terms of the target state's Schmidt coefficients, the domain can be relaxed to be convex as any pure state entanglement measure is a concave function, i.e., $E(\Psi^*) \geq \theta E(\Psi')+ (1-\theta) E(\Psi'') \geq E_{\rm target}$ for $\lambda_i^\downarrow(\Psi^*) = \theta \lambda_i^\downarrow(\Psi') + (1-\theta) \lambda_i^\downarrow(\Psi'')$ with $E(\Psi') \geq E_{\rm target}$, $E(\Psi'') \geq E_{\rm target}$ and  $0 \leq \theta \leq 1$.. However, the target function $\displaystyle \min_{l \in \{ 1,2, \cdots, d\} } \left( \frac{\sum_{i=l}^d \lambda_i^\downarrow(\Psi)} {\sum_{i=l}^d \lambda_i^\downarrow(\Psi')} \right)$ is neither convex nor concave, so that the overall optimisation does not belong to a convex optimisation problem.
\bibitem{note2} Random Schmidt coefficients (brown circle points) are sampled from a uniformly distributed $d$-dimensional vector. Haar random $d \times d$ pure bipartite states (orange rectangular points) are sampled using the Haar measure on the unitary group $U(d^2)$. In order to get a broad entanglement distribution of initial states, we also sample random states (purple diamond points) by choosing a random number between $\lambda_1 \in [0,1]$ then randomly choose $\lambda_2 \in [0, 1- \lambda_1]$, followed by a random choice of $\lambda_3 \in [0, 1-\lambda_1-\lambda_2]$, and so on. $10^4$ random states are sampled by each method.
\end{thebibliography}
\end{document}